\documentclass[preprint,12pt]{elsarticle}

\usepackage[margin=1in]{geometry}

\usepackage{graphicx}

\usepackage{amssymb}

\usepackage{url}

\usepackage{amsmath}
\usepackage{upgreek}

\usepackage{amsthm}

\def\doubleunderline#1{\underline{\underline{#1}}}
\usepackage{tensor}
\usepackage{dsfont}

\usepackage{siunitx}

\usepackage{amssymb}

\usepackage{bbm}

\usepackage{soul}

\usepackage{chemformula}

\usepackage{algorithmicx}
\usepackage{algorithm}
\usepackage[noend]{algpseudocode}
\makeatletter
\renewcommand{\ALG@beginalgorithmic}{\footnotesize}
\makeatother

\usepackage{booktabs}

\usepackage{amsmath}

\makeatletter
\newcommand{\vast}{\bBigg@{4}}
\newcommand{\Vast}{\bBigg@{7}}
\makeatother

\usepackage{caption}

\usepackage{subcaption}

\setcitestyle{authoryear}
 \biboptions{comma,round}

\journal{Journal of Computational Physics}

\begin{document}

\begin{frontmatter}

\title{Modeling transport of scalars in two-phase flows with\\ a diffuse-interface method}


\author{Suhas S. Jain\corref{cor1}}
\ead{sjsuresh@stanford.edu}
\cortext[cor1]{Corresponding author}
\author{Ali Mani}
\ead{alimani@stanford.edu}
\address{Center for Turbulence Research, Stanford University, California, USA 94305}

\begin{abstract}
In this article, we propose a novel scalar-transport model for the simulation of scalar quantities in two-phase flows with a phase-field method (diffuse-interface method). In a two-phase flow, the scalar quantities typically have disparate properties in two phases, which results in effective confinement of the scalar quantities in one of the phases, in the time scales of interest. This confinement of the scalars lead to the formation of sharp gradients of the scalar concentration values at the interface, presenting a serious challenge for its numerical simulations. 

To overcome this challenge, we propose a model for the transport of scalars. 
The model is discretized using a central-difference scheme, which leads to a non-dissipative implementation that is crucial for the simulation of turbulent flows. Furthermore, the provable strengths of the proposed model are: (a) the model maintains the positivity property of the scalar concentration field, a physical realizability requirement for the simulation of scalars, when the proposed criterion is satisfied, (b) the proposed model is such that the transport of the scalar concentration field is consistent with the transport of the volume fraction field, which results in the enforcement of the effective zero-flux boundary condition for the scalar at the interface; and therefore, prevents the artificial numerical diffusion of the scalar across the interface. 

Finally, we present numerical simulations using the proposed model in a wide range of two-phase flow regimes, spanning laminar to turbulent flows; and assess: the accuracy and robustness of the model, the validity of the positivity property of the scalar concentration field, and the enforcement of the zero-flux boundary condition for the scalar at the interface.
\end{abstract}


\end{frontmatter}



\section{Introduction \label{sec:intro}}


The transport of scalars in a two-phase flow is an important problem that finds applications in wide range of natural phenomena and industrial processes. A scalar quantity can represent: temperature field in modeling boiling and evaporation phenomena \citep{villegas2016ghost},
dissolved gas concentration in modeling oceanic carbon sequestration process \citep{lal2008carbon}, salt concentration in modeling electrochemical systems \citep{fernandez2014bubble},
surfactant concentration in modeling Marangoni effects \citep{takagi2011surfactant}, etc.

The transport of scalars in a two-phase flow, typically, involves very disparate length and time scales; and the scalar quantities often experience very large and small diffusivities and mobilities in different phases. For example, the diffusion coefficient of \ch{CO2} in air is $1.6\times10^{-5}\ \mathrm{m^2/s}$; whereas, in water it is $1.6\times10^{-9}\ \mathrm{m^2/s}$, which is $\approx10,000$ times smaller than in air. Similarly, the thermal diffusivity of air, at room temperature, is $1.9\times10^{-5}\ \mathrm{m^2/s}$; whereas, that of water is $1.43\times10^{-7}\ \mathrm{m^2/s}$, which is $\approx100$ times smaller compared to that of air. In a worst-case scenario, the diffusivities in the two phases could be so different that the ratio of diffusivities in two phases could tend to infinity. These disparate properties of the scalar quantities in different phases result in scalars effectively being confined to one of the phases in the time scales of interest. This poses a numerically challenging task to resolve the gradient of the scalar concentration at the material interface, and it usually leads to numerical leakage (artificial numerical diffusion) at the interface and negative values of the scalar concentration field. 

Several studies in the past have quantified, and tackled this issue of numerical leakage in the context of a volume-of-fluid (VOF) method. \citet{alke2009vof} and \citet{bothe2013volume} used a piecewise linear interface calculation (PLIC) algorithm for the transport of species at the interface to prevent artificial mass transfer across the interface. \citet{hassanvand2012direct} demonstrated that with no special treatment for the scalar-transport equation, the scalar quantity artificially diffuses across the interface into the other phase, even though the diffusivity of the scalar in the other phase is set to zero. They proposed a fix to the issue by modifying the flux of the scalar and by setting it to zero at the interface. \citet{berry2013multiphase} combined the zero-flux boundary condition for charged ions with the ion-transport equation and derived a modified transport equation that implicitly achieves the ion-impenetrable boundary condition at the interface.  


In this study, however, we use a diffuse-interface method for accurate modeling of the interface between two fluids. Advantages of a diffuse-interface method over a geometric VOF method is the absence of geometric reconstruction of the interface that results in a low-cost, load-balanced, and a highly-scalable method, see \citet{jain2020conservative} for the parallel-scalability performance of a diffuse-interface method on large scale computing machines. However, a VOF method is known to be more accurate compared to a diffuse-interface method on grids of same size. Therefore, on a per-cost basis, both VOF and diffuse-interface methods are known to perform similarly for incompressible two-phase flows \citep{mirjalili2019comparison}. 
Compared to a level-set method, a diffuse-interface method also maintains discrete conservation of volume of each of the phases. For a more detailed discussion and a comparison of various interface-capturing methods, see the recent review by \citet{Mirjalili2017}. 




Unlike the numerous advances in modeling the transport of scalars in the context of a VOF method, there have been very few studies for scalar transport in a diffuse-interface method. Scalar equation algorithm (SEA) is a method that is similar to a diffuse-interface method, wherein the transport of volume fraction is done using the finite-difference schemes with no geometric reconstruction. \citet{pericleous1995free} used the SEA method for modeling free-surface flows and adopted the Van Leer total-variation diminishing (TVD) scheme for the discretization of the energy equation instead of an upwinding scheme to avoid the artificial heat loss across the interface due to numerical smearing. However, the mere use of a Van Leer TVD scheme for the discretization of the scalar-transport equation does not completely eliminate the artificial diffusion of the scalar across the interface \citep{mehdi2004modelling}. Most VOF methods tackle the issue of artificial numerical diffusion of the scalar across the interface by geometrically advecting the scalar concentration field in a manner consistent with the advection of the volume fraction field \citep{davidson2002volume,alke2009vof,bothe2013volume,berry2013multiphase}. However, in a diffuse-interface method, advection of the volume fraction field is performed non-geometrically with the use of finite-difference schemes, which is a main contributing factor in reducing the cost and improving the scalability of the method compared to a VOF method. But the lack of geometric fluxing in diffuse-interface method poses a challenge towards modeling scalars and to prevent the artificial numerical diffusion (numerical leakage) of the scalar across the interface.



To the best of our knowledge, still lacking are numerical methods with simulation capabilities that can accurately capture the behavior of the scalar concentration fields without any numerical leakage from one phase to the other while still maintaining the positivity of the scalar concentration values. For example, in the work of \citet{davidson2002volume}, negative values of the scalar were reset to zero to prevent the method from diverging. To address these deficiencies, we have developed a general scalar-transport model for two-phase flows, particularly for material interfaces modeled using a diffuse-interface method. The proposed model is general enough that the scalar can represent soluble surfactants, electrolytes, temperature, chemical species in combustion modeling, etc. We adopt the two-scalar approach \citep{davidson2002volume,alke2009vof,bothe2013volume,berry2013multiphase,ma2013numerical} where a separate scalar-transport equation is solved for each of the phases but each of the scalar-transport equations are solved in the entirety of the domain. Our newly developed scalar-transport model prevents the artificial numerical diffusion of scalars from one phase to the other, even in challenging flow environments such as in the presence of electrokinetic effects and in turbulent flows. We discretize the proposed model equation using a second-order central scheme in space due to its non-dissipative nature, that is crucial for the simulation of turbulent flows \citep{mittal1997suitability}. We have proved that the resulting discrete equation maintains the positivity of the scalar concentration field throughout the simulation, provided the grid resolution is sufficient to resolve the length scales present in the flow. We have thus derived a ``positivity criterion" that determines this length scale that needs to be resolved to maintain the positivity of the scalar concentration field for two-phase flows. 



The rest of this paper is organized as follows. The conservative diffuse-interface method, that is used for modeling two-phase flows in this work, is presented in Section \ref{sec:diffuse-interface}. The discrete representation of the scalar concentration field adopted in this work is formally presented in Section \ref{sec:transport} along with some of the previous scalar-transport models available in the literature. The newly proposed scalar-transport model in this work is presented in Section \ref{sec:new_model}; followed by the presentation of the proof of positivity for the scalar concentration field, and the definition of the positivity criterion in terms of the grid resolution requirements, in Section \ref{sec:positivity}. Finally, the adopted numerical strategy is summarized in Section \ref{sec:numerics}; and results from the numerical simulations are presented in Section \ref{sec:results}, illustrating the applicability of the proposed scalar-transport model for modeling passive and active scalars in a wide range of two-phase flow settings spanning laminar to turbulent flow regimes. The summary of results and conclusions are presented in Section \ref{sec:conclusion}.

\section{Conservative phase-field method\label{sec:diffuse-interface}}

The first step towards simulating scalars in a two-phase flow is to choose an interface-capturing method that can accurately simulate interfaces in challenging flow environments. In this work, we choose a phase-field for the reasons described in the previous section. Within the class of phase-field methods, the popular approaches are based on the Cahn-Hilliard equation by \citet{cahn1958free} and the Allen-Cahn equation by \citet{allen1979microscopic} which were originally proposed to describe the phase separation and coarsening phenomena in solids and the motion of antiphase boundaries in crystalline solids, respectively. More recently, these equations have been successfully used in fluid dynamics to model the transport of material interfaces between two fluids \citep{anderson1998diffuse,lowengrub1998quasi,chen1998applications,jacqmin1999calculation,liu2003phase,badalassi2003computation,yue2004diffuse,yang2006numerical,kim2012phase}. However, there are some downsides to the use of these equations. The Cahn-Hilliard equation conserves the volume of each phase but contains a fourth-order derivative term which requires special care for its discretization. The Allen-Cahn equation, on the other hand, does not conserve the volume of each phase but involves a second-order derivative term. To have the advantages of both equations, \citet{chiu2011conservative} developed the conservative phase-field method for incompressible two-phase flows by combining the reinitialization step of the conservative level-set method by \citet{olsson2005conservative} with the phase-field equation. 

In this work, we choose to use the conservative phase-field method for incompressible two-phase flows. In this method the phase-field equation is written as 
\begin{equation}
\frac{\partial \phi}{\partial t} + \vec{\nabla}\cdot(\vec{u}\phi) = \vec{\nabla}\cdot\Big[\Gamma\Big\{\epsilon\vec{\nabla}\phi - \phi(1 - \phi)\vec{n}\Big\}\Big],
\label{eq:phi}
\end{equation}
where $\phi$ is the phase-field variable which represents volume fraction of the given phase, $\Gamma\{\epsilon\vec{\nabla}\phi - \phi(1 - \phi)\vec{n}\}=\vec{a}(\phi)$ term on the right-hand side is the flux of the interface regularization (diffusion-sharpening) term, $\vec{n}=\vec{\nabla}\phi/|\vec{\nabla}\phi|$ is the unit normal vector to the interface, and $\Gamma$ and $\epsilon$ are the interface parameters, where $\Gamma$ represents an artificial regularization velocity scale and $\epsilon$ represents an interface thickness scale. This equation satisfies both $\phi_1$ and $\phi_2$, where $\phi_1$ and $\phi_2=1-\phi_1$ are the phase-field variables for phases $1$ and $2$, respectively. We can show that the interface regularization term satisfies $\vec{a}(\phi_1)=-\vec{a}(\phi_2)$ for phases $1$ and $2$. 
The phase-field equation [Eq. \eqref{eq:phi}] can be solved in conjunction with the momentum balance equation
\begin{equation}
\frac{\partial \rho\vec{u}}{\partial t} + \vec{\nabla}\cdot(\rho \vec{u} \otimes \vec{u} + p \mathds{1}) = \vec{\nabla}\cdot\doubleunderline\tau + \vec{\nabla}\cdot(\vec{S}\otimes\vec{u}),
\label{eq:mod_mom}
\end{equation}
where
\begin{equation}
    \vec{S}=\Gamma\left[\epsilon \vec{\nabla} \rho - \frac{(\rho_1 - \rho)(\rho - \rho_2)}{\rho_1 - \rho_2} \vec{n}\right]
\end{equation}
is the net mass regularization flux, $\rho_l$ is the density of the phase $l$ and is a constant, and $\rho=\sum_{l=1}^2 \rho_l\phi_l$ is the total density of the mixture. Here, $\vec{S}$ is added to the incompressible momentum equation so that the resulting momentum is transported in a manner consistent with the transport of volume fraction using the phase-field equation \citep{mirjalili2020consistent}. A similar consistency corrections for the momentum equation have been proposed for the Cahn-Hilliard and Allen-Cahn models in \citet{huang2020cahn} and \citet{huang2020allen}, respectively. The Cauchy stress tensor is written as $\doubleunderline\tau = 2\mu\mathbb{D} - 2\mu(\vec{\nabla}\cdot\vec{u})\mathds{1}/3$, where $\mu$ is the dynamic viscosity of the mixture evaluated using the one-fluid mixture rule \citep{kataoka1986local} as $\mu=\sum_{l=1}^2 \phi_l \mu_l$, $\mu_l$ is the dynamic viscosity of the phase $l$, $\mathbb{D}=\{(\vec{\nabla}\vec{u}) + (\vec{\nabla}\vec{u})^T\}/2$ is the strain-rate tensor, and $p$ is the pressure. In this work, both the phases are assumed to be incompressible. Therefore, a pressure field $p$ that satisfies the criterion of zero divergence of the velocity field can be computed using the fractional-step method of \citet{kim1985application}. However, this does not limit the applicability of the newly proposed scalar-transport model in this work for incompressible flows; it can also be used with the conservative diffuse-interface method proposed by \citet{jain2020conservative} for compressible two-phase flows because of the same form of the equilibrium kernel function\textemdash a hyperbolic-tangent function\textemdash for the phase field $\phi$ (see, Section \ref{sec:equilibrium}).


It can be shown that the phase field $\phi$ is bounded between $0$ and $1$ \citep{Mirjalili2020}, provided $\Gamma$ and $\epsilon$ are chosen such that they satisfy the criterion 
\begin{equation}
    \frac{\epsilon}{\Delta x} \ge \frac{\left(\frac{|u|_{\mathrm{max}}}{\Gamma} + 1\right)}{2},
    \label{eq:DIcrossover}
\end{equation}
where $\Delta x$ is the grid size and $|u|_{\mathrm{max}}$ is the maximum value of the magnitude of the velocity in the domain. This method also inherently satisfies the TVD property for the phase field \citep{jain2020conservative} with the use of a central-difference scheme for the discretization of all the operators and without having to add any flux limiters that destroy the non-dissipative nature of the scheme.

\section{Problem description\label{sec:transport}}

\begin{figure}
    \centering
    \includegraphics[width=0.4\textwidth]{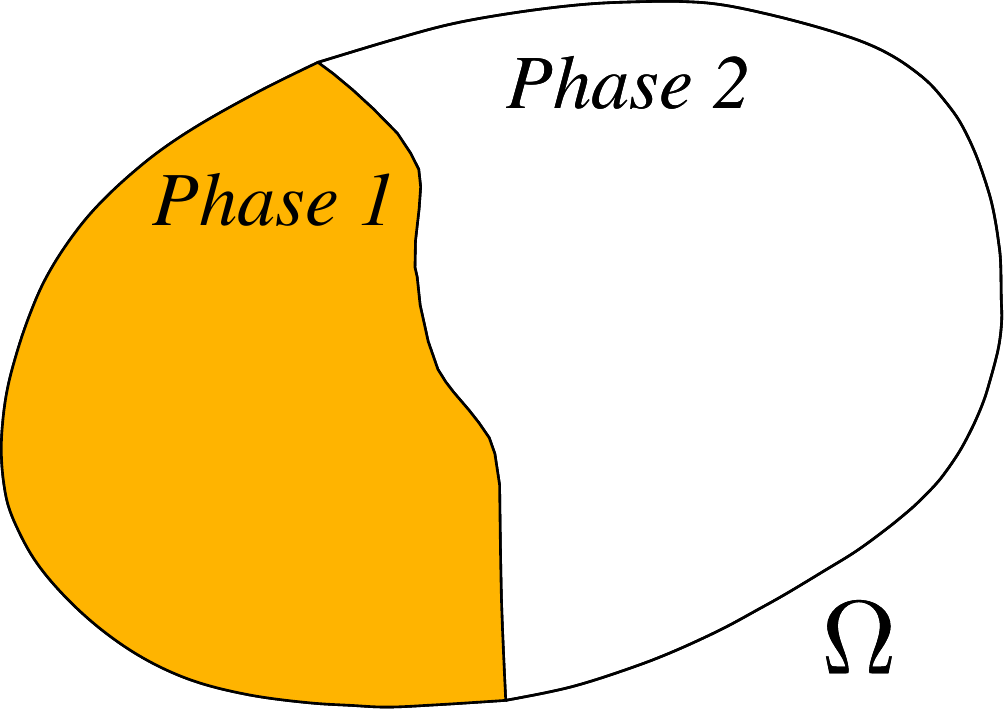}
    \caption{Schematic representing the domain filled with two phases and a scalar quantity that is confined to phase $1$. Color represents the scalar concentration.}
    \label{fig:domain}
\end{figure}
Consider the schematic of a domain, $\Omega$, with two phases $1$ and $2$ shown in Figure \ref{fig:domain}. Let $c$ represent the concentration (amount of scalar per unit volume) of a scalar quantity in the domain. The scalar is said to be conserved, $\int_{\Omega} c\  dV = \mathrm{constant}$, if it is not being generated or destroyed. As described in Section \ref{sec:intro} the ratio of diffusivites of the scalar in two phases is typically very large, hence the scalar will be confined to one of the phases as illustrated in the Figure \ref{fig:domain}. Now reformulating and generalizing the problem at hand, i.e., the diffusivity is, say, a finite value $D$ in phase $1$ and practically zero in phase $2$, then the scalar is confined only to phase $1$ in the domain. 

Since the scalar is confined to phase $1$, we can then define another variable $\tilde{c}$ that represents the local concentration of the scalar in phase $1$ (amount of scalar per unit volume of the phase $1$). Then, the relation between $c$ and $\tilde{c}$ is
\begin{equation}
    c = \phi \tilde{c},
    \label{eq:relation}
\end{equation}
where $\phi$ represents the volume fraction of the phase (volume of the phase per unit total volume) where the scalar is present (phase $1$ in this case). By definition, the value of $\tilde{c}$ is undefined in phase $2$. In the continuum limit of an infinitely sharp interface, $\phi$ reduces to a Heaviside function with a value of one inside phase $1$ and zero inside phase $2$. Therefore, in this limit, the scalar concentration can be represented as
\begin{equation}
    c = \Bigg\{
    \begin{aligned}
        & \tilde{c} \hspace{10mm} \mathrm{inside\ the\ phase\ where\ scalar\ is\ present}\\
        & 0 \hspace{10mm} \mathrm{elsewhere}.
    \end{aligned}
\label{eq:c_sharp}
\end{equation}

Now, assuming that $\tilde{c}$ satisfies a generic transport (advection-diffusion) equation within phase $1$, an evolution equation for $\tilde{c}$ can be written as
\begin{equation}
\frac{\partial \tilde{c}}{\partial t} + \vec{\nabla}\cdot(\vec{u}_c\tilde{c}) = \vec{\nabla} \cdot (D\vec{\nabla} \tilde{c}),
\label{equ:c_tilde}
\end{equation}
where $\vec{u}_c=\vec{u}+\vec{u}_r$ represents the total convective velocity of the scalar, $\vec{u}$ represents the fluid velocity, $\vec{u}_r$ represents any effective velocity with which the scalar is being advected relative to the fluid (e.g., electromigration velocity $\vec{u}_r= \nu \vec{E}$, where $\nu$ is the electrical mobility and $\vec{E}$ is the electric field), and $D$ represents the diffusivity of the scalar. Since the scalar is confined to phase $1$ with the interface acting as the boundary, this represents a classical boundary value problem. The appropriate no-flux boundary condition for this problem can be written as
\begin{equation}
(D\nabla\tilde{c} - \vec{u}_r\tilde{c})\cdot\vec{n}=0.
\label{equ:c_bc}
\end{equation}
Combining Eq. \eqref{equ:c_tilde} and Eq. \eqref{equ:c_bc}, \citet{berry2013multiphase} derived the modified transport equation for charged ions that can implicitly achieve the ion-impenetrable boundary condition as
\begin{equation}
\frac{\partial (\phi \tilde{c})}{\partial t} + \vec{\nabla}\cdot(\phi \vec{u}_c\tilde{c}) = \vec{\nabla} \cdot (D \phi \vec{\nabla} \tilde{c}). 
\label{equ:c_tilde_phi}
\end{equation}
They used this modified transport equation along with a sharp-interface method (coupled level-set and volume-of-fluid method) to simulate the electrokinetics of liquid-liquid systems.


In the discrete limit when the interface is not represented as a sharp boundary\textemdash more so with the use of diffuse-interface methods\textemdash however, the interface is not infinitely sharp; and the smallest interface thickness that can be handled on an Eulerian grid is of the order of grid-cell size. In this limit, Eq. \eqref{eq:c_sharp} would not hold anymore, which is illustrated in the schematic shown in Figure \ref{fig:scalar_def}.
\begin{figure}
    \centering
    \includegraphics[width=\textwidth]{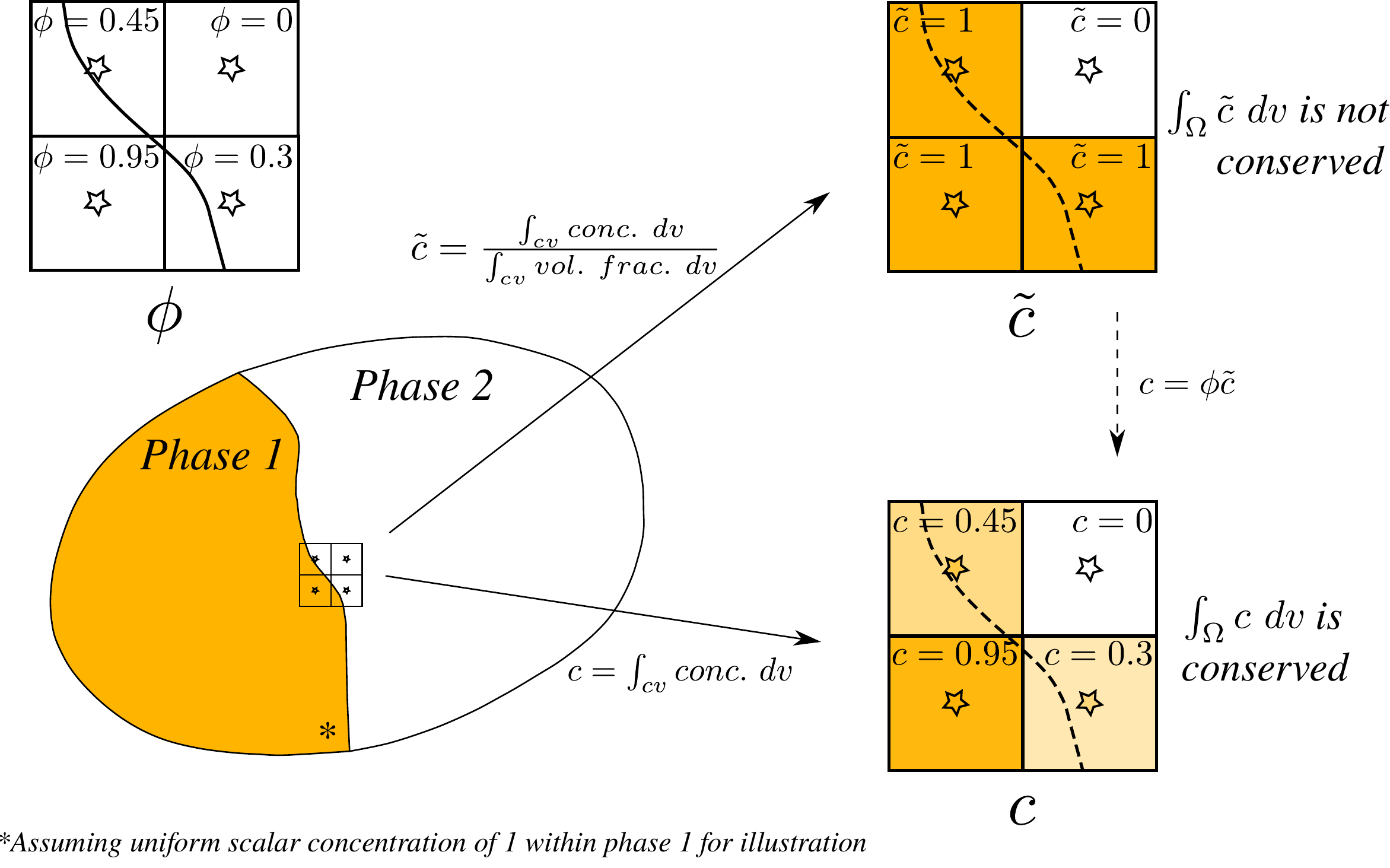}
    \caption{Schematic representing the discrete $c$ and $\tilde{c}$ quantities. Opacity of the color is proportional to the field value in that particular cell. Here, ``conc." represents the local scalar concentration, ``vol. frac." is the local volume fraction of the phase and $cv$ represents a control volume.}
    \label{fig:scalar_def}
\end{figure}
In this discrete limit and with the use of finite-difference schemes, i.e., without any special geometric treatment, we can observe that $\tilde{c}$ is not a conserved quantity, i.e., $\int_{\Omega} \tilde{c}\  dV \ne \mathrm{constant}$, because of the values that it takes in the cells that contain the interface. This is the reason most VOF methods use geometric fluxing for the transport of the scalar quantities. However, $\int_{\Omega} c\  dV = \mathrm{constant}$ still holds, and therefore, $c$ is still a conserved quantity in this discrete limit. Hence, we could instead write a transport equation for $c$ and look for modifications that mimic the no-flux boundary condition at the interface. In other words, we seek a transport equation for $c$ that results in transport of the scalar quantity consistent with the transport of the phase field (see Figure \ref{fig:equilibrium_solution}; we define consistency more rigorously in Section \ref{sec:equilibrium}), such that there is no artificial numerical diffusion of the scalar at the interface. 

With this notion, we could start with a generic form of the transport equation for $c$ as
\begin{equation}
\frac{\partial c}{\partial t} + \vec{\nabla}\cdot(\vec{u}_c c) = \vec{\nabla} \cdot (D\vec{\nabla} c).
\label{equ:c}
\end{equation}
A straightforward modification (a naive approach) to this equation could be to multiply $\phi$ to $\vec{u}_r$ and $D$, such that the flux $D\vec{\nabla}c-\vec{u}_r c$ goes to zero as $\phi$ goes to zero outside phase $1$ as
\begin{equation}
\frac{\partial c}{\partial t} + \vec{\nabla}\cdot( \vec{u} c + \phi \vec{u}_r c) = \vec{\nabla} \cdot (D \phi \vec{\nabla} c).
\label{equ:c_mod}
\end{equation}
Though this simple modification to achieve no-flux boundary condition sounds promising, the scalar concentration field $c$ that we obtain by solving Eq. \eqref{equ:c_mod} will not be consistent with the phase field$\phi$ and it results in the artificial numerical diffusion across the interface, especially for diffuse-interface methods. However, this approach has been used previously with a VOF method in the context of modeling heat and mass transfer across interfaces \citep{davidson2002volume}. 

Figure \ref{fig:fix} shows the application of Eq. \ref{equ:c_mod} for the evolution of the scalar quantity in a stationary one-dimensional drop. Since the initial concentration $c$ of the scalar was uniform everywhere within the drop, the initial state of $c$ should be maintained throughout the simulation. However, the model in Eq. \ref{equ:c_mod} results in predicting a lower value of $c$ with an $\mathcal{O}(1)$ error due to the numerical leakage of the scalar across the interface as illustrated in Figure \ref{fig:fix} (b), thus affecting the overall accuracy of the solution. For problems with a two-way coupling between the flow and the transport of scalar, the local concentration of the scalar field is crucial in predicting accurate flow fields. For example, the transport of reacting species in combustion modeling and ion transport in electrokinetics both modify the flow field depending on the local concentration of the scalar being transported. Therefore, the spurious leakage of the scalar and the resulting $\mathcal{O}(1)$ error in these cases, could be detrimental to the overall accuracy of the simulation. In the worst-case scenario, the numerical solution for the scalar concentration field can admit unphysical negative values near the interface and often within the phase with the lower diffusivity value, which might result in unrealizable scalar concentration fields. 
\begin{figure}
\vskip 0.1in
\centering
\includegraphics[width=\textwidth]{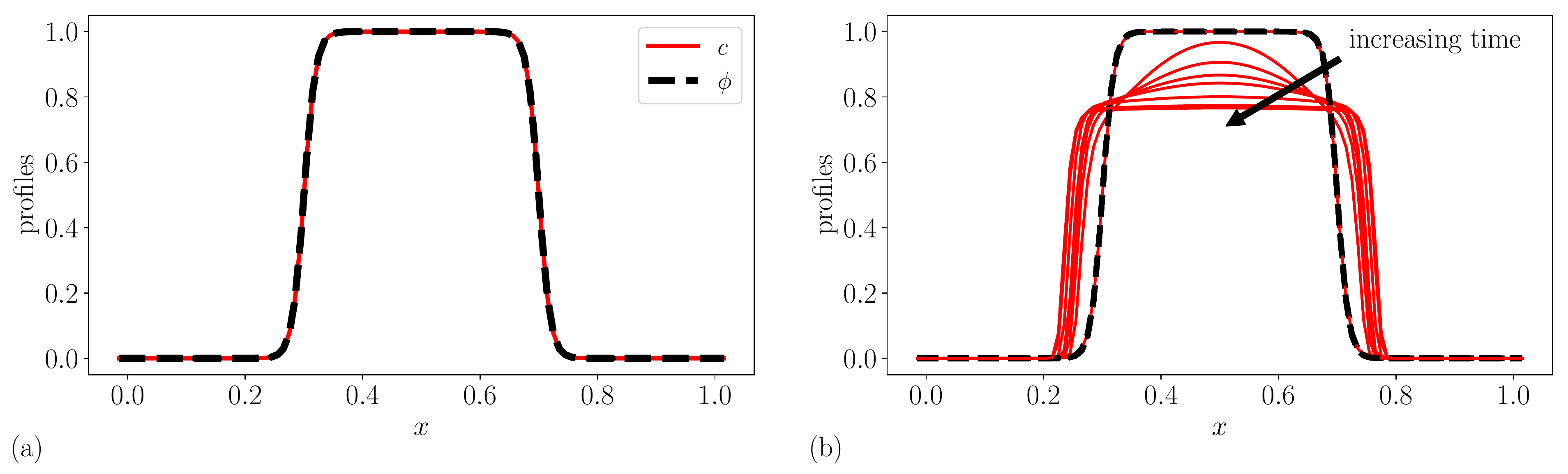}
\caption{Simulation results obtained by solving Eq. \eqref{equ:c_mod}, showing the one-dimensional drop ($\phi$ field) and the scalar concentration ($c$ field) at (a) initial time and (b) later times. Parameters used for this simulation are $D=0.01$, $\Delta x = 0.01$, $u=0$, and $u_r=0$. The arrow denotes the time evolution of the scalar concentration field.}
\label{fig:fix}
\end{figure}

One trivial approach to address this issue of unphysical negative values is to reset them to zero every time step \citep{davidson2002volume}. However, this would destroy the conservative property of the scalar quantity. Other common and more sophisticated ways to address this issue is to use flux limiters and positivity-preserving limiters \citep{laney1998computational}. However, the numerical diffusivity associated with these schemes inevitably adds to the unphysical leakage of the scalar across the interface into the impermeable phase. Hence, in the current work, we present a new model equation for the transport of scalars in two-phase flows along with a consistent numerical discretization scheme that overcomes the challenges that were presented here. 

\section{Proposed model equation for the transport of scalars in two-phase flows\label{sec:new_model}}

With the objective of developing a scalar-transport model that does not admit negative values for the scalar concentration field, nor permit unphysical leakage of the scalar across the interface, we propose a scalar-transport model for the scalar quantity that is confined to phase $1$ as
\begin{equation}
\frac{\partial c}{\partial t} + \vec{\nabla}\cdot(\vec{u} c + \phi \vec{u}_r c) = \vec{\nabla} \cdot \left[D \left\{\vec{\nabla}c - \frac{(1 - \phi) \vec{n} c}{\epsilon} \right\}\right],
\label{equ:interface}
\end{equation}
where $\vec{n}$ is the normal vector to the interface, $\epsilon$ is the same interface parameter that is present in Eq. \eqref{eq:phi}. The proposed model equation in Eq. \eqref{equ:interface} is a modification to Eq. \eqref{equ:c_mod}, which failed to prevent the numerical leakage of the scalar and maintain the no-flux boundary condition for the scalar at the interface (Figure \ref{fig:fix}). Away from the interface, in the bulk region of phase $1$ where the scalar is present ($\phi\rightarrow1$), the proposed model in Eq. \eqref{equ:interface} reduces to
\begin{equation}
\frac{\partial c}{\partial t} + \vec{\nabla}\cdot(\vec{u} c + \vec{u}_r c) = \vec{\nabla} \cdot \big(D\vec{\nabla}c\big),
\label{equ:bulk}
\end{equation}
which is the transport equation for scalars in a single-phase flow. Hence, the model has no adverse effect on the transport of scalars away from the interface.

Since we adopt the two-scalar approach, wherein a separate scalar equation is written for each of the phases, an equation for the scalar in phase $2$, represented by the phase field $1-\phi$, can be written by replacing $\phi$ with $1-\phi$ in Eq. \ref{equ:interface} as
\begin{equation}
\frac{\partial c'}{\partial t} + \vec{\nabla}\cdot\{\vec{u} c' + (1 - \phi) \vec{u}_r c'\} = \vec{\nabla} \cdot \left\{D \left(\vec{\nabla}c + \frac{\phi \vec{n} c'}{\epsilon} \right)\right\},
\label{equ:interface2}
\end{equation} 
where $c'$ represents the concentration of the scalar quantity that is confined to phase 2. Furthermore, the transfer of scalar across the interface can also be straightforwardly accounted for in this two-scalar approach by introducing sink/source terms  \citep{davidson2002volume} into the scalar-transport model in Eqs. \eqref{equ:interface},\eqref{equ:interface2} and will be explored in a future work.

\subsection{Consistency and equilibrium solution\label{sec:equilibrium}}

As described in Section \ref{sec:transport}, the consistency of the scalar concentration $c$ and the phase field $\phi$ is crucial in preventing the unphysical numerical leakage of the scalar across the interface. The $c$ and $\phi$ are formally said to be \textit{consistent} if they possess the same equilibrium kernel function\textemdash a hyperbolic tangent function\textemdash such that the local concentration field $\tilde{c}$, which is defined as the ratio $\tilde{c}=c/\phi$, stays constant across the interface. The local concentration field $\tilde{c}$ then represents the interfacial concentration of the scalar quantity. The proposed scalar-transport model in Eq. \eqref{equ:interface} was derived such that $c$ satisfies this consistency condition.

To realize the above defined consistency between $c$ and $\phi$, we can look at the steady-state equilibrium solutions to the phase-field equation and the proposed scalar-transport equation. At steady state, for $\vec{u}=0$ and $\vec{u}_r=0$ and in one dimension, the phase-field equation in Eq. \eqref{eq:phi} reduces to the form
\begin{equation}
0 = \vec{\nabla} \cdot \left[\Gamma\epsilon \left\{\vec{\nabla}\phi - \frac{(1 - \phi) \vec{n} \phi}{\epsilon} \right\}\right] \Rightarrow \frac{d^2 \phi}{d x^2} - \frac{1}{\epsilon} \frac{d\left\{(1-\phi)\phi\right\}}{d x} = 0,
\label{equ:phi_steady}
\end{equation}
for $\vec{n}=+1$, and the proposed scalar-transport equation in Eq. \eqref{equ:interface} reduces to the form
\begin{equation}
0 = \vec{\nabla} \cdot \left[D \left\{\vec{\nabla}c - \frac{(1 - \phi) \vec{n} c}{\epsilon} \right\}\right] \Rightarrow \frac{d^2 c}{d x^2} - \frac{1}{\epsilon} \frac{d\left\{(1-\phi)c\right\}}{d x} = 0.
\label{equ:c_steady}
\end{equation}

\begin{figure}
    \centering
    \includegraphics[width=0.6\textwidth]{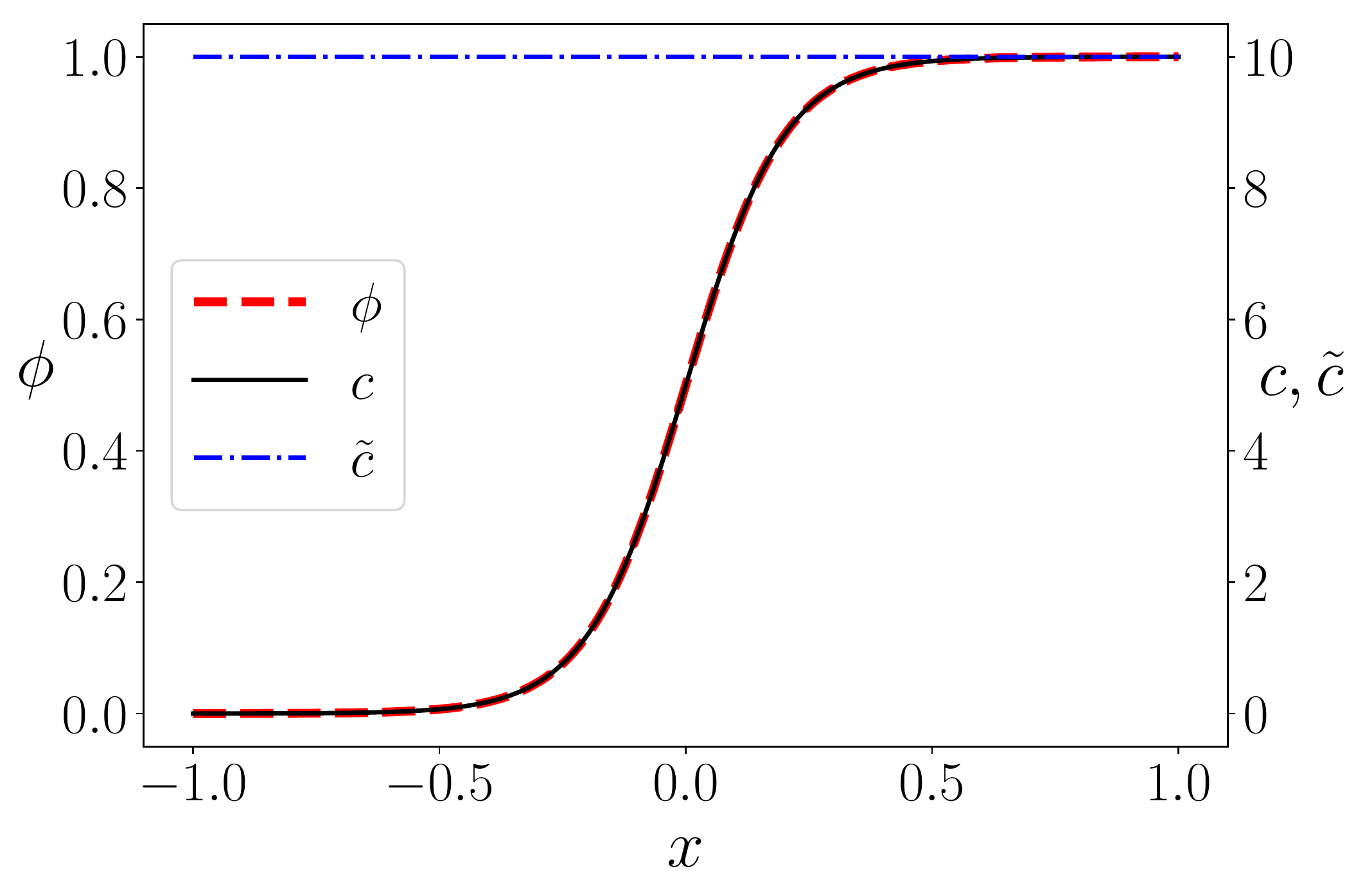}
    \caption{Equilibrium solutions for $\phi$ and $c$, showing consistency between the scalar concentration field and the volume fraction field, therefore, a constant value of the local concentration $\tilde{c}=c/\phi$. Here, $c_0$ is chosen to be equal to $5$ for the sake of illustration.}
    \label{fig:equilibrium_solution}
\end{figure}

Now, assuming that the interface is at the origin as shown in Figure \ref{fig:equilibrium_solution}, and integrating Eq. \ref{equ:phi_steady} along with the boundary conditions
\begin{equation}
    \phi = \Bigg\{
    \begin{aligned}
        & 0 \hspace{10mm} x \rightarrow -\infty \\
        & 0.5 \hspace{10mm} x = 0,
    \end{aligned}
\end{equation}
we obtain
\begin{equation}
\phi = \frac{e^{(x/\epsilon)}}{1 + e^{(x/\epsilon)}} = \frac{1}{2} \left\{1 + \tanh{\left(\frac{x}{2\epsilon}\right)}\right\}.
\label{equ:phi_sol}
\end{equation}
Using the equilibrium solution for $\phi$ in Eq. \eqref{equ:phi_sol} and solving for $c$ by integrating the Eq. \eqref{equ:c_steady} and using the boundary conditions 
\begin{equation}
    c = \Bigg\{
    \begin{aligned}
        & 0 \hspace{10mm} x \rightarrow -\infty \\
        & c_0 \hspace{10mm} x = 0,
    \end{aligned}
\end{equation}
we obtain
\begin{equation}
c = 2c_0\frac{e^{(x/\epsilon)}}{1 + e^{(x/\epsilon)}} = c_0 \left\{1 + \tanh{\left(\frac{x}{2\epsilon}\right)}\right\}.
\label{equ:c_sol}
\end{equation}
Hence, the equilibrium kernel functions of $c$ and $\phi$ are both hyperbolic tangent functions of the spatial coordinate along the interface normal (here $x$), and therefore, are consistent at the interface. This results in a constant value of $\tilde{c}=c/\phi=2c_0$ across the interface (see Figure \ref{fig:scalar_def}). If we choose $c_0$ to be equal to $0.5$, the equilibrium kernel functions of $c$ and $\phi$ are identical, and the resulting interfacial scalar concentration is $\tilde{c}=1$. For other values of $c_0$, $c$ is only scaled by a constant factor, but is still a hyperbolic tangent function as illustrated in Figure \ref{fig:equilibrium_solution}. Hence, the scalar concentration field obtained from solving the proposed model in Eq. \eqref{equ:interface} is consistent with the phase field from Eq. \eqref{eq:phi}, which results in the transport of the scalar inside the confined phase without any unphysical numerical leakage across the interface that would be seen otherwise.  

\section{Positivity of scalars\label{sec:positivity}}

Positivity of the scalar concentration field is a crucial realizability requirement that needs to be satisfied at all times in the simulation. A common approach to achieve this is to use flux limiters and positivity-preserving limiters \citep{laney1998computational}. However, these limiters add artificial numerical dissipation to the scheme and also lead to unphysical numerical leakage of the scalar across the interface. Hence, we use a central-difference scheme to discretize the operators in our system of equations because of its well-known non-dissipative property \citep{moin2016suitability}, and derive a criterion [Eq. \eqref{eq:crossover}; Figure \ref{fig:grid_resolution}] for the choice of grid size to be used in the simulation. We propose in Theorem \ref{theorem:positivity} that the criterion in Eq. \eqref{eq:crossover} is a sufficient condition for the proposed model transport equation in Eq. \eqref{equ:interface} to maintain the positivity of the scalar concentration field $c$ at all times during the simulation.

\newtheorem{theorem}{Theorem}[section]

\begin{theorem}
If $\phi^k_i$ is bounded between $0$ and $1$, $\forall k\in\mathds{Z}^+$ and $\forall i$, on a uniform one-dimensional grid, then $c^k_i\ge 0$ holds $\forall k\in\mathds{Z}^+$, where $k$ is the time-step index and $i$ is the grid index, provided
\begin{equation}
    \Delta x \le \left(\frac{2 D}{|u|_{\mathrm{max}} + |u_r|_{\mathrm{max}} + \frac{D}{\epsilon}}\right),
    \label{eq:crossover}
\end{equation}
and
\begin{equation}
      \Delta t \le \frac{\Delta x ^2}{2D}
    \label{eq:boundtime}
\end{equation}
are satisfied, where $\Delta x$ is the grid-cell size, $\Delta t$ is the time-step size, $|u|_{\mathrm{max}}$ and $|u_r|_{\mathrm{max}}$ are the maximum fluid velocity and the maximum relative velocity of the scalar in the domain, respectively.

\label{theorem:positivity}
\end{theorem}


\begin{proof}
Consider the discretization of Eq. \eqref{equ:interface} on a one-dimensional uniform grid
\begin{equation}
\begin{aligned}
c^{k+1}_i = c^k_i + \Delta t \left[- \left(\frac{u^k_{i+1}c^k_{i+1} - u^k_{i-1}c^k_{i-1}}{2\Delta x} \right) -\left(\frac{{u^k_{r}}_{i+1}\phi^k_{i+1}c^k_{i+1} - {u^k_{r}}_{i-1}\phi^k_{i-1}c^k_{i-1}}{2\Delta x} \right)\right] \\
+ \Delta t \left[D\left(\frac{c_{i+1}^k -2c_{i}^k + c_{i-1}^k}{\Delta x^2}\right) - \frac{D}{\epsilon}\left\{\frac{(1 - \phi_{i+1}^k) n_{i+1}^k c_{i+1}^k - (1 - \phi_{i-1}^k) n_{i-1}^k c_{i-1}^k}{2\Delta x}\right\} \right].
\end{aligned}
\end{equation}
This can be rearranged as
\begin{equation}
c_i^{k+1} = \tilde{C}_{i-1}^k c_{i-1}^k + \tilde{C}_{i}^k c_{i}^k + \tilde{C}_{i+1}^k c_{i+1}^k,
\end{equation}
where $\tilde{C}$'s are coefficients given by
\begin{equation}
    \tilde{C}_{i-1}^k = \frac{\Delta t u^k_{i-1}}{2\Delta x} + \frac{\Delta t {u^k_r}_{i-1} \phi^k_{i-1}}{2\Delta x} + \frac{\Delta t D}{\Delta x^2} + \frac{\Delta t D}{2\epsilon\Delta x}(1 - \phi_{i-1}^k) n_{i-1}^k,
    \label{eq:ci-1}
\end{equation}
\begin{equation}
    \tilde{C}_{i+1}^k = -\frac{\Delta t u^k_{i+1}}{2\Delta x} -\frac{\Delta t {u^k_r}_{i+1}\phi^k_{i+1}}{2\Delta x} + \frac{\Delta t D}{\Delta x^2} - \frac{\Delta t D}{2\epsilon\Delta x}(1 - \phi_{i+1}^k) n_{i+1}^k
    \label{eq:ci+1}
\end{equation}
and
\begin{equation}
    \tilde{C}^k_i=1 - \frac{2 \Delta t D}{\Delta x^2}.
    \label{eq:ci}
\end{equation}
\newtheorem{lemma}{Lemma}[theorem]
\begin{lemma}
A scheme is said to maintain positivity (also called the ``boundedness" criterion in \citet{patankar1980numerical,versteeg2007introduction}) if $\tilde{C}$'s are all positive \citep{laney1998computational}.
\label{lemma:positive}
\end{lemma}

It is given that $0\le\phi^k_i\le1$ holds $\forall k\in\mathds{Z}^+$, which implies that $(1 - \phi_{i-1}^k) n_{i-1}^k \ge -1$. Using this in Eq. \eqref{eq:ci-1}, we obtain
\begin{equation}
    \begin{aligned}
    \tilde{C}_{i-1}^k \ge \frac{\Delta t u^k_{i-1}}{2\Delta x} + \frac{\Delta t {u^k_r}_{i-1} \phi^0_{i-1}}{2\Delta x} + \frac{\Delta t D}{\Delta x^2} - \frac{\Delta t D}{2\epsilon\Delta x}\\
    \Rightarrow \tilde{C}_{i-1}^k \ge -\frac{\Delta t }{2\Delta x} \bigg(|u|^k_{max} + |u_r|^k_{max} + \frac{D}{\epsilon}\bigg) + \frac{\Delta t D}{\Delta x^2}
    \end{aligned}
\end{equation}
Now, invoking the condition in Eq. (\ref{eq:crossover}), we can show that $\tilde{C}_{i-1}^k \ge 0$ holds. Using similar arguments, we can show that $\tilde{C}_{i+1}^k \ge 0$ holds. Invoking the condition in Eq. (\ref{eq:boundtime}), we can also show that $\tilde{C}_{i}^k \ge 0$ holds. Thus, Lemma \ref{lemma:positive} proves that $c^k_i\ge0$ is satisfied $\forall k\in\mathds{Z}^+$, which concludes the proof.

\end{proof}

Now, generalizing Theorem \ref{theorem:positivity} for three dimensions, the time-step restriction required for the positivity of $c$\textemdash assuming an isotropic mesh\textemdash can be written as
\begin{equation}
      \Delta t \le \frac{\Delta x ^2}{6D}.
    \label{eq:3Dboundtime}
\end{equation}
Note that $\phi$ is assumed to be bounded between $0$ and $1$ throughout the simulation. This boundedness of $\phi$ can be achieved by choosing appropriate values for the interface parameters $\Gamma$ and $\epsilon$ such that they satisfy the criterion in Eq. \eqref{eq:DIcrossover}. For more details on the boundedness of $\phi$, see \citet{Mirjalili2020} for incompressible flows and \citet{jain2020conservative} for compressible flows.

The criterion on the time-step size in Eq. \eqref{eq:boundtime} and Eq. \eqref{eq:3Dboundtime} are the Courant-Friedrich-Levy (CFL) conditions for the scalar diffusion process in one and three dimensions, respectively, and are typically already satisfied in an explicit time-marching scheme to achieve temporal stability. Hence, the only additional criterion that needs to be satisfied to maintain the positivity of the evolution of the scalar concentration field is the restriction on the grid size given in Eq. \eqref{eq:crossover}. In the proof of Theorem \ref{theorem:positivity}, a first-order Euler time-stepping scheme was used to arrive at the restrictions on the time-step size in Eq. \eqref{eq:boundtime} and Eq. \eqref{eq:3Dboundtime}; however, these criteria are sufficient to maintain the positivity of $c$ with most higher-order explicit time-stepping schemes since the diffusive CFL conditions for the scalar in Eq. \eqref{eq:boundtime} and Eq. \eqref{eq:3Dboundtime} are less restrictive for higher-order time-stepping schemes.

It was shown in Section \ref{sec:new_model} that the proposed model transport equation in Eq. \eqref{equ:interface} reduces to a generic scalar-transport equation [Eq. \eqref{equ:bulk}] for single-phase flows, in the bulk region away from the interface. Now, repeating the analysis in Theorem \ref{theorem:positivity} for the scalar-transport equation for the bulk region [Eq. \eqref{equ:bulk}], we can show that the criterion that needs to be satisfied to maintain the positivity of evolution of the scalar concentration field is the restriction on the grid size given by
\begin{equation}
        \Delta x \le \bigg(\frac{2 D}{|u|_{\mathrm{max}} + |u_r|_{\mathrm{max}}}\bigg).
        \label{eq:bulkcrossover}
\end{equation}

A similar analysis holds good if one is interested in the large-eddy simulation (LES) of scalars in a two-phase turbulent flow, instead of a direct numerical simulation (DNS). Here, the diffusivity $D$ of the scalar can be replaced by the effective diffusivity (sum of resolved and subgrid contributions) of the scalar in Eq. \eqref{eq:crossover} and Eq. \eqref{eq:bulkcrossover}.


\subsection{Spatial resolution requirements}

\begin{figure}
    \centering
    \includegraphics[width=0.6\textwidth]{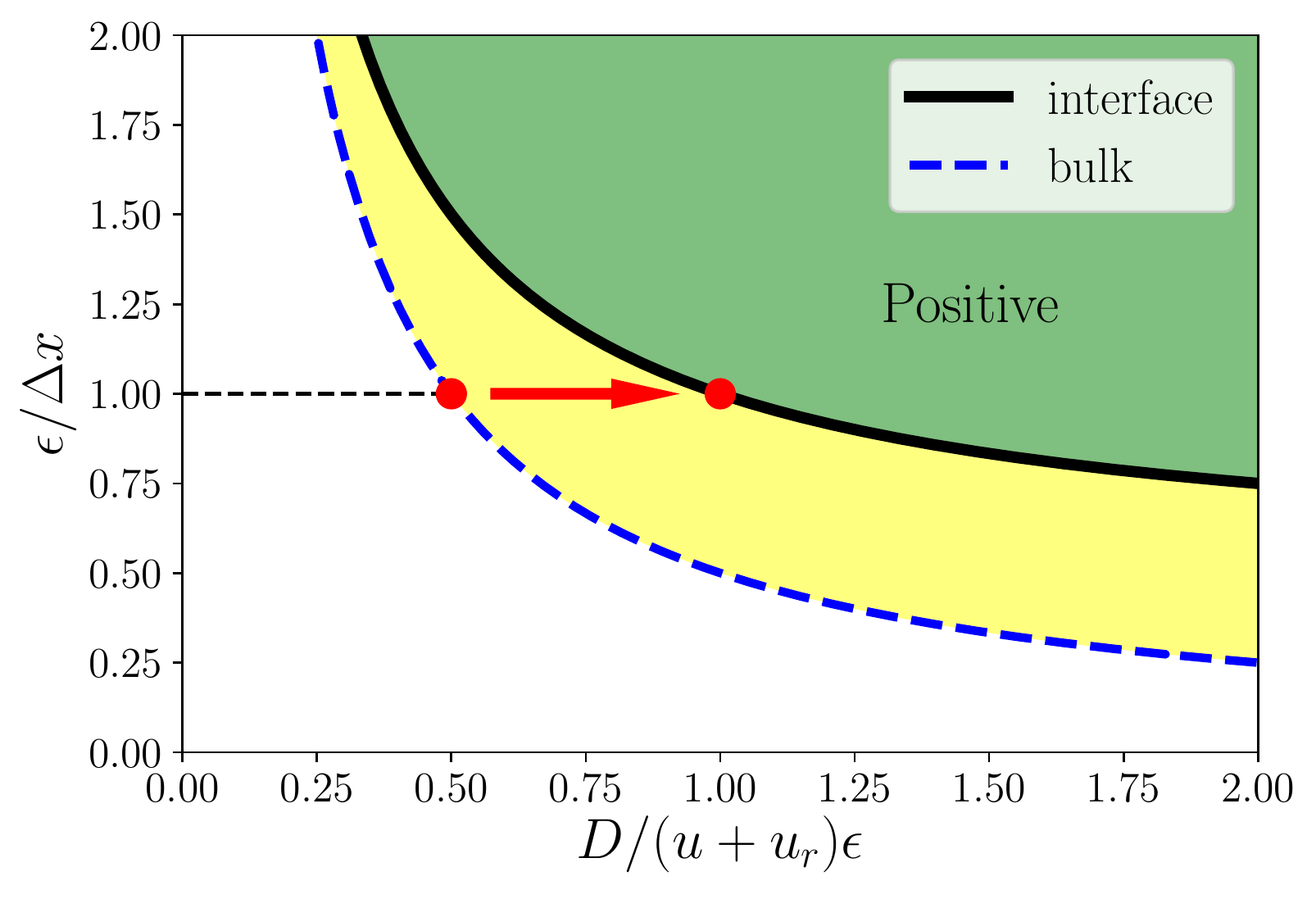}
    \caption{Graphical representation of the positivity criteria. The solid line is the positivity criterion in Eq. \eqref{equ:nondim_crossover} for the proposed scalar-transport model [Eq. \eqref{equ:interface}], and the dashed line is the positivity criterion in Eq. \eqref{equ:nondim_bulkcrossover} for the scalar-transport model in the bulk region away from the interface [Eq. \eqref{equ:bulk}]. The red dots represent the positivity criteria for the choice of $\epsilon=\Delta x$ and the arrow represents the additional restriction imposed on the grid size at the interface compared to the bulk region away from the interface.}
    \label{fig:grid_resolution}
\end{figure}

Now, rewriting the positivity criterion for the proposed scalar-transport model in Eq. \eqref{eq:crossover} in terms of the non-dimensional groups $[\epsilon/\Delta x]$ and $[D/\{(u + u_r)\epsilon\}]$ as
\begin{equation}
    \Big[\frac{\epsilon}{\Delta x}\Big] \ge \frac{1 + \Big[\frac{D}{(u + u_r)\epsilon}\Big]}{2\Big[\frac{D}{(u+u_r)\epsilon}\Big]}
    \label{equ:nondim_crossover}
\end{equation}
and the positivity criterion in Eq. \eqref{eq:bulkcrossover} for the bulk region away from the interface as
\begin{equation}
    \Big[\frac{\epsilon}{\Delta x}\Big] \ge \frac{1}{2\Big[\frac{D}{(u+u_r)\epsilon}\Big]},
    \label{equ:nondim_bulkcrossover}
\end{equation}
where $u=|u|_{max}$ and $u_r=|u_r|_{max}$, symbol $[\cdot]$ represents a non-dimensional group. By plotting $[\epsilon/\Delta x]$ vs $[D/\{(u + u_r)\epsilon\}]$ in Figure \ref{fig:grid_resolution}, we can see that any grid size $\Delta x$ for a given $\epsilon$ and $D$ that lies above the lines maintains the positivity of the scalar concentration field throughout the simulation. Since the solid line is above the dashed line in Figure \ref{fig:grid_resolution} for all values of $[D/\{(u + u_r)\epsilon\}]$, the proposed model transport equation in the full form in Eq. \eqref{equ:interface} imposes a more restrictive condition on the grid size $\Delta x$ for a given $\epsilon$ and $D$ than the reduced form of the transport equation for the bulk region in Eq. \eqref{equ:bulk}. Therefore, the criterion in Eq. \eqref{eq:crossover} should be used to choose the grid size throughout the domain.


In the conservative phase-field method, $\epsilon$ is typically chosen to be equal to $\Delta x$. With this choice of $\epsilon$, and recognizing that the non-dimensional group $[D/\{(u + u_r)\epsilon\}] = [D/\{(u + u_r)\Delta x\}]$ is the inverse of cell-Peclet number $Pe_c$, the positivity criterion for the proposed model in Eq. \eqref{eq:crossover} can be written in terms of cell-Peclet number as
\begin{equation}
    1 \ge \frac{1 + \Big[\frac{1}{Pe_c}\Big]}{2\Big[\frac{1}{Pe_c}\Big]} \Rightarrow Pe_c \le 1,
    \label{equ:pecletcriterion}
\end{equation}
and the positivity criterion in Eq. \eqref{eq:bulkcrossover} for the bulk region away from the interface can be written as
\begin{equation}
    1 \ge \frac{1}{2\Big[\frac{1}{Pe_c}\Big]} \Rightarrow Pe_c \le 2.
    \label{equ:single_pecletcriterion}
\end{equation}
This criterion in Eq. \eqref{equ:single_pecletcriterion} was first proposed by \citet{patankar1980numerical} and \citet{versteeg2007introduction} for a generic scalar-transport equation for single-phase flows. Hence, the presence of a material interface and the effective no-flux boundary condition for the scalar that prevents it from artificially diffusing across the interface has introduced a more restrictive criterion on the grid size in terms of the cell-Peclet number as $Pe_c\le1$, which otherwise would have been $Pe_c\le2$ in the absence of any interface. This is also graphically shown (arrow) in Figure \ref{fig:grid_resolution}. Therefore, assuming a uniform grid throughout the domain, the grid size for the simulation of scalars with material interfaces (two-phase flow) should be twice as small compared to the grid size for the simulation of scalars in the absence of material interfaces (single-phase flow). One could also use a twice-refined grid only around the interface using an adaptive-mesh refinement (AMR) technique instead of using a uniform grid throughout the domain, to reduce the cost of computation.

The restriction on the grid size in Eq. \eqref{equ:pecletcriterion} can also be motivated by the physical scales involved in the problem. A cell-Peclet number of $Pe_c\le1$, implies that the grid size is $\Delta x \le D/(u+u_r)$. Here, $D/(u+u_r)\sim l_c$ represents the characteristic length scale $l_c$ of the problem. Therefore the cell-Peclet number restriction is alluding to the fact that the grid size should be small enough to resolve the physical length scales present in the problem, i.e., $\Delta x \lesssim l_c$.

\section{Numerical strategy \label{sec:numerics}}

The scalar-transport model proposed in this work is implemented in the CTR-DIs2D and CTR-DIs3D solvers \citep{Jain2018cons,jain2020conservative}. These solvers can handle both compressible and incompressible flows. For incompressible flows, a finite-volume discretization strategy on a staggered grid has been employed wherein the phase field, the pressure field, and the scalar concentration fields are stored at the cell centers; and the components of the velocity field vector are stored at the cell faces where all the fluxes are evaluated. This choice of discretization is adopted, for incompressible flows, to avoid the spurious checkerboarding of the pressure field \citep{patankar1980numerical}. The pressure-Poisson equation is solved with a geometric-multigrid preconditioned conjugate gradient method using the HYPRE package \citep{falgout2002hypre}. 


We use the fourth-order Runge-Kutta (RK4) time-stepping scheme and the second-order central-differencing scheme for the discretization of the spatial operators. This choice of numerical scheme has some advantages, particularly for the simulation of turbulent flows due to its (a) non-dissipative nature, (b) low aliasing error, (c) easy boundary treatment, (d) low cost, and (e) improved stability \citep{moin2016suitability}. With the appropriate choice of $\Delta x$, $\Gamma$, and $\epsilon$, we can achieve the positivity for the scalar concentration field (see, Section \ref{sec:positivity}) and the boundedness and TVD properties for the phase field (see, Section \ref{sec:diffuse-interface}) even with the use of a central-difference scheme for the spatial operators, which would otherwise admit oscillatory solutions due to the associated dispersion errors. 

\section{Simulation results \label{sec:results}}

The proposed scalar-transport model can be used with a wide range of two-phase flows, from laminar to turbulent flow regimes. To illustrate this, multiple test cases are presented in this section, starting from the simple one-dimensional cases of a droplet advection with a scalar in Section \ref{sec:cases}. These cases were used to assess the validity of the positivity criterion of the proposed model by testing the model for various choices of parameters on the positivity map in Figure \ref{fig:grid_resolution}. This is followed by the two-dimensional cases of a bubble in a concentration-polarization region in Section \ref{sec:scalar_bubble} which involves modeling the transport of a scalar quantity around the bubble, and charged ions in a drop in Section \ref{sec:ion_drop}. To simulate charged ions, the proposed scalar-transport model was recast into a Nernst-Planck equation and was solved in conjunction with the Gauss's law. Finally, a three-dimensional case of a scalar in a turbulent two-phase flow is presented in Section \ref{sec:scalar_turbulent}. In all the above cases, the scalar is confined to one of the phases, and therefore, the ability of the model to prevent the unphysical numerical leakage of the scalar across the interface is also evaluated.


\subsection{Advection of a droplet along with a scalar\label{sec:cases}}


In this section, a one-dimensional drop of radius $R=0.25$ is advected in a periodic domain along with a scalar quantity that is confined to the drop, to test the validity of the positivity criterion in Eq. \eqref{equ:pecletcriterion}. Four different choices of the parameters are made as shown in Figure \ref{fig:cases}(a) in terms of the cell-Peclet number ($Pe_c=(u+u_r)\Delta x/D$), such that the two out of the four choices ($Pe_c=1,0.8$) satisfy the positivity criterion and the other two ($Pe_c=2,4$) violate the criterion. This test case is repeated with a zero and a non-zero relative velocity of the scalar with respect to the fluid, and with a non-uniform initial conditions for the scalar concentration field as further described in the subsequent Sections \ref{sec:zerorelvel}, \ref{sec:nonzerorelvel}, and \ref{sec:nonuniform}, respectively. For all the simulations in this section, the drop is initially placed in the domain centered at $x_0=0.5$ along with the scalar, and then advected with a uniform prescribed velocity. The scalar has a uniform concentration of $1$ inside the drop and $0$ outside, as shown in Figure \ref{fig:cases}(b), unless specified otherwise. The domain is discretized into a uniform grid of size $\Delta x=0.01$. The interface parameters are chosen as $\Gamma=100$ and $\epsilon=\Delta x$. 

\begin{figure}
    \centering
    \includegraphics[width=0.49\textwidth]{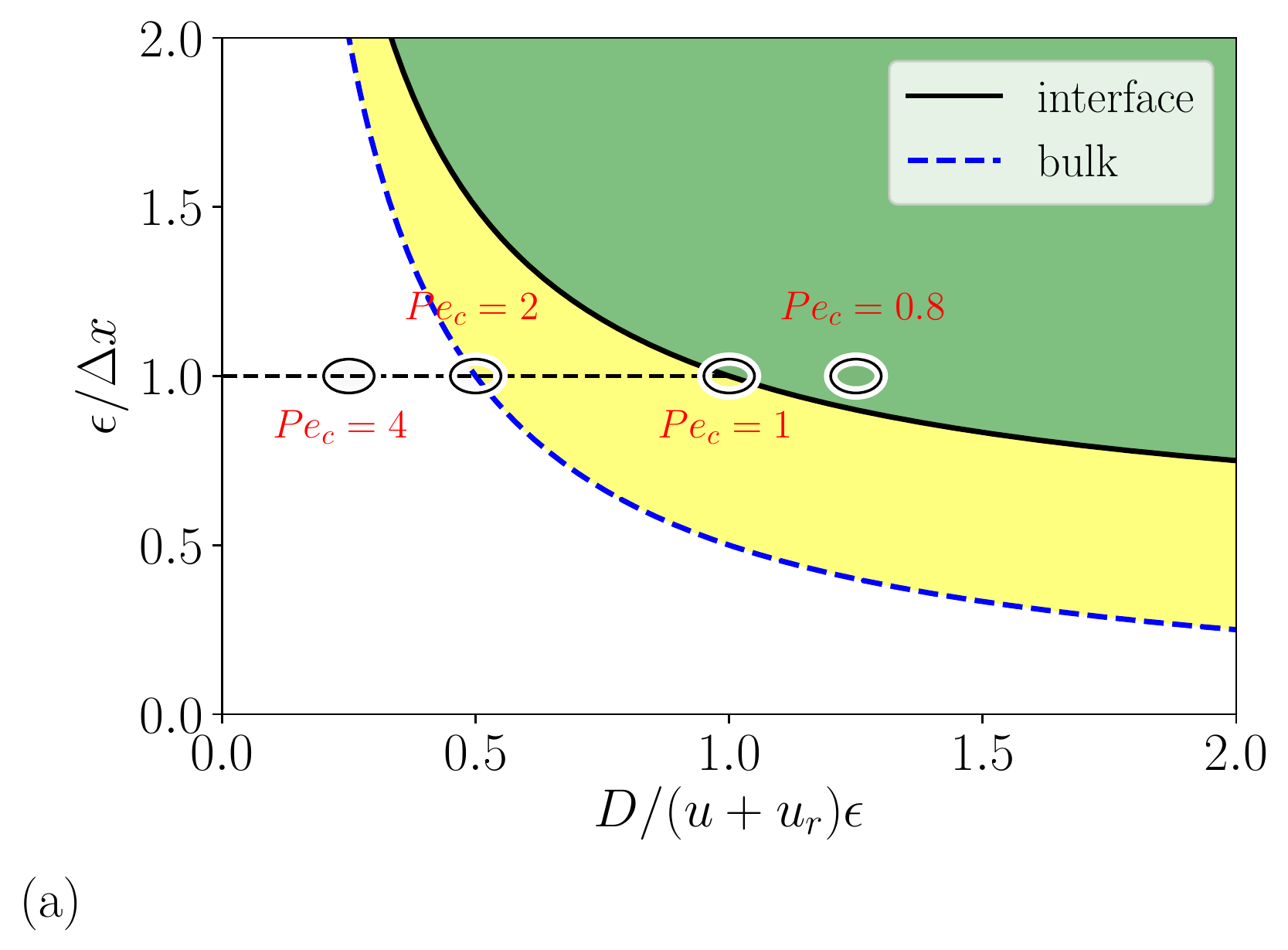}
    \includegraphics[width=0.49\textwidth]{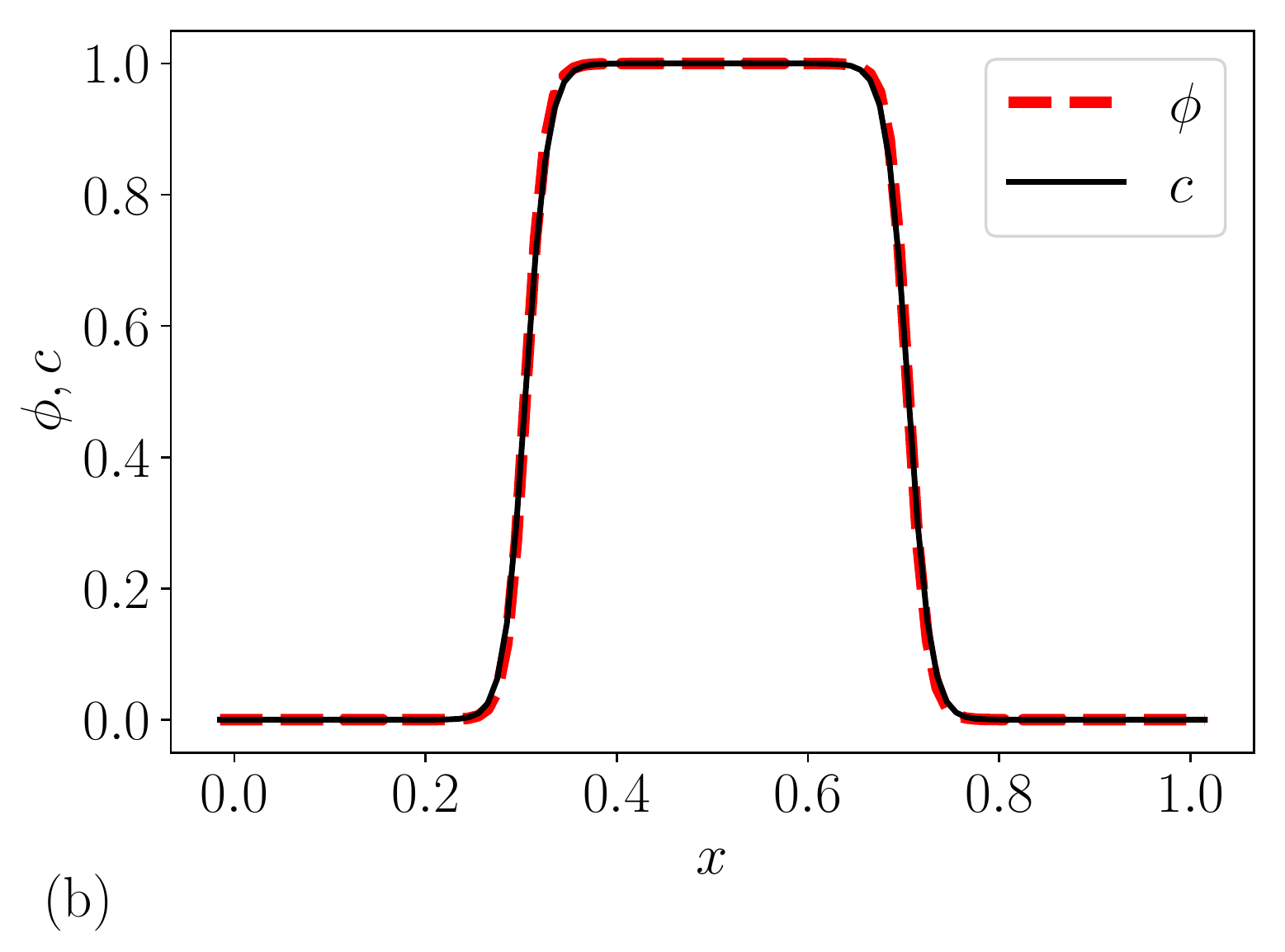}
    \caption{(a) Four different choices of the parameters in terms of the cell-Peclet number $Pe_c$, represented by ovals on the positivity plot, that are selected to test the validity of the positivity criterion. (b) Initial conditions for the volume fraction field $\phi$ and the scalar concentration field $c$.}
    \label{fig:cases}
\end{figure}

\subsubsection{Scalar advection along with the drop ($\vec{u}_r=0$)\label{sec:zerorelvel}}

Here in this section, the uniform fluid velocity is chosen to be $\vec{u}=100$, the relative velocity of the scalar with respect to the fluid is $\vec{u}_r=0$, the domain range is $[0,1]$, and the total integration time is $t_{\mathrm{end}}=10$. The diffusivity of the scalar is chosen based on the required cell-Peclet number. The drop advects to the right due to the imposed fluid velocity and returns to its original position at $t=0.01$ due to the periodic boundary condition, and this process repeats $1000$ times until the time $t=10$. The final state of the drop and the scalar concentration field at time $t=10$ is shown in Figure \ref{fig:cases14} along with the minimum value of the scalar concentration field for all four choices of the parameters shown in Figure \ref{fig:cases}(a). For the cases with $Pe_c=1$ and $0.8$, that satisfy the positivity criterion in Eq. \eqref{equ:pecletcriterion}, the scalar concentration $c$ is positive throughout the domain; and for the cases with $Pe_c=2$ and $4$, that violate the positivity criterion, the scalar concentration $c$ admits negative values close to the interface, as expected. 

\begin{figure}
    \centering
    \includegraphics[width=\textwidth]{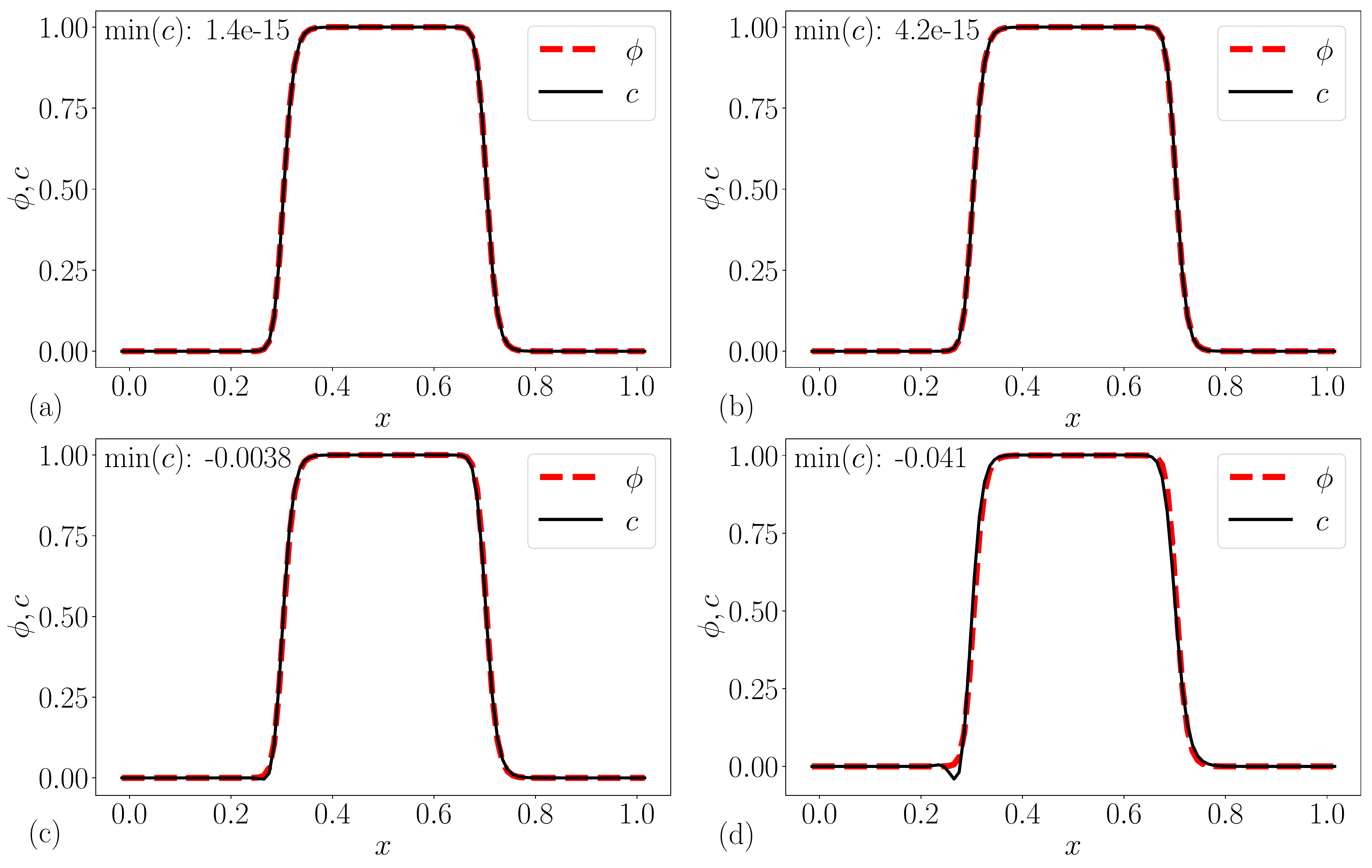}
    \caption{Final state of the drop and the scalar concentration field at time $t=10$ for the case of scalar advection along with the drop.  Four plots represent the four different parameters chosen to test the positivity of the scalar: (a) $Pe_c=1$, (b) $Pe_c=0.8$, (c) $Pe_c=2$, and (d) $Pe_c=4$.}
    \label{fig:cases14}
\end{figure}

Comparing the results in Figure \ref{fig:cases14} with those in Figure \ref{fig:fix}, where the scalar was found to artificially leak outside the drop resulting in an $\mathcal{O}(1)$ error for the scalar concentration values, it is easy to see the role and importance of the proposed scalar-transport model in maintaining the consistency between $c$ and $\phi$ and in the prevention of any artificial numerical leakage of the scalar. To further quantify the error and inconsistency between $c$ and $\phi$, $c-\phi$ and the local concentration $\tilde{c}$ are plotted in Figure \ref{fig:cases14error} for the same four choices of the parameters shown in Figure \ref{fig:cases}(a). Since $c$ is expected to be same as $\phi$ in this case, we could use $c-\phi$ as a metric to evaluate how far the solution deviates compared to the expected one. As can be seen in Figure \ref{fig:cases14error}, the solution is exact for $Pe_c=1$, but there is a small error for $Pe_c=0.8$ case. The error also seems to increase with the increasing value of $Pe_c$ for $Pe_c>1$. Small non-zero values of $c-\phi$ close to the interface imply that the scalar concentration field is out-of-equilibrium by a small amount. A similar behavior to $c-\phi$ can be seen in the plots of $\tilde{c}$. Here, $\tilde{c}$ is computed as 
\begin{equation}
\tilde{c} = \Bigg\{ 
    \begin{aligned}
        & \frac{c}{\phi} & \hspace{10mm}  \phi > 0.01 \\
        & 0 & \hspace{10mm} \mathrm{else},
    \end{aligned}
 \label{}   
\end{equation}
to avoid division by a small number.

\begin{figure}
    \centering
    \includegraphics[width=\textwidth]{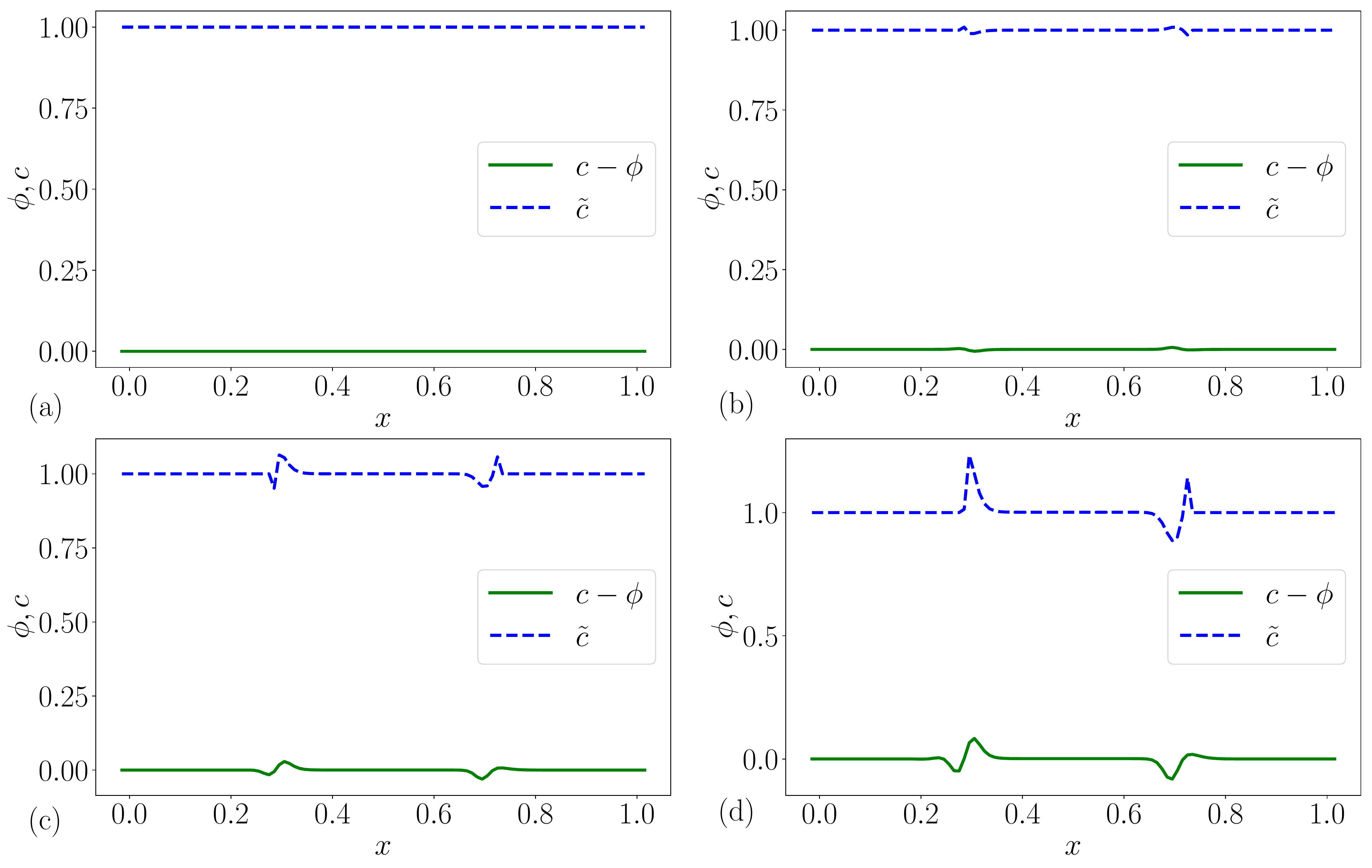}
    \caption{$c-\phi$ and $\tilde{c}$ at the final time of $t=10$ for the case of scalar advection along with the drop.  Four plots represent the four different parameters: (a) $Pe_c=1$, (b) $Pe_c=0.8$, (c) $Pe_c=2$, and (d) $Pe_c=4$.}
    \label{fig:cases14error}
\end{figure}

\subsubsection{Scalar advection relative to the drop ($\vec{u}_r\ne 0$)\label{sec:nonzerorelvel}}

Here, the uniform fluid velocity is chosen to be $\vec{u}=50$, the relative velocity of the scalar with respect to the fluid is $\vec{u}_r=50$, the domain range is $[0,1]$, and the total integration time is $t_{\mathrm{end}}=10$. The diffusivity of the scalar is chosen based on the required cell-Peclet number. The drop advects to the right due to the imposed fluid velocity and returns to its original position at $t=0.02$ due to the periodic boundary condition, and this process repeats $500$ times until the end of the simulation. The scalar also advects to the right, relative to the drop, and accumulates on the right end of the drop since it is confined to the drop. It reaches a steady state when the diffusion balances the advection due to the relative velocity. The final state of the drop and the scalar concentration field at time $t=10$ are shown in Figure \ref{fig:cases58} along with the minimum values of the scalar concentration field for all four choices of the parameters shown in Figure \ref{fig:cases}. For the cases with $Pe_c=1$ and $0.8$, that satisfy the positivity criterion, the scalar concentration $c$ is positive throughout the domain; and for the case with $Pe_c=4$, that violate the positivity criterion, the scalar concentration $c$ admits negative values for $Pe_c=4$, as expected. However the positivity is still maintained for $Pe_c=2$ though the criterion is violated. This is because the criterion is only a sufficient, not a necessary condition for the scalar concentration field to remain positive and therefore presents the most restrictive condition such that the positivity of the scalar is satisfied even in some situations when the criterion is violated. This shows the robustness of the positivity criterion that is crucial in maintaining the realizable values of the scalar concentration field throughout the duration of the simulation.
\begin{figure}
    \centering
    \includegraphics[width=\textwidth]{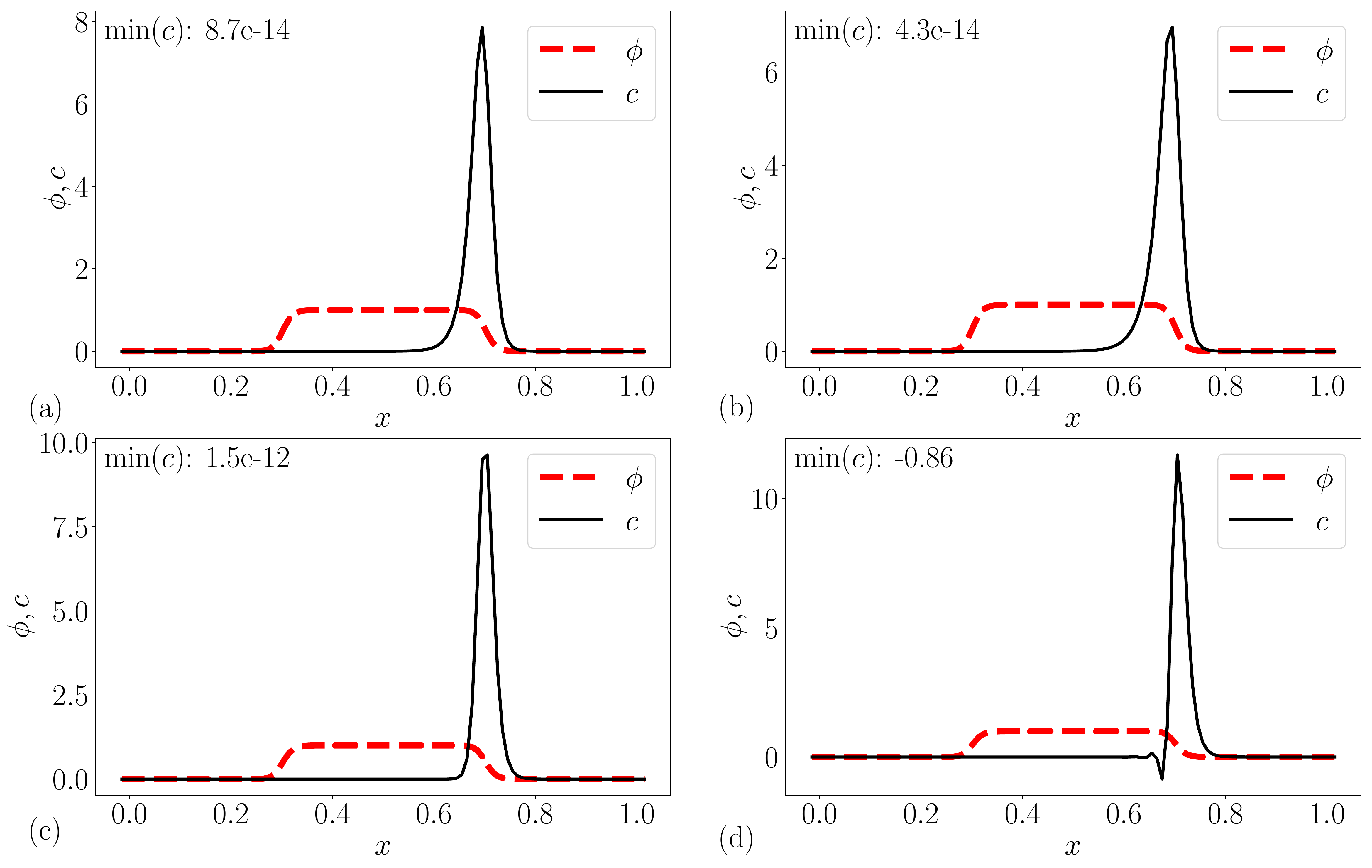}
    \caption{Final state of the drop and the scalar concentration field at time $t=10$ for the case of scalar advection relative to the drop. Four plots represent the four different parameters chosen to test the positivity of the scalar: (a) $Pe_c=1$, (b) $Pe_c=0.8$, (c) $Pe_c=2$, and (d) $Pe_c=4$ with $u_r\ne0$.}
    \label{fig:cases58}
\end{figure}

\subsubsection{Non-uniform scalar concentration within the drop\label{sec:nonuniform}}

The simulations in the previous Sections \ref{sec:zerorelvel} and \ref{sec:nonzerorelvel} had uniform initial scalar concentration fields within the drop. However, in reality, the scalar concentration could be non-uniform inside the drop. To evaluate the model for this scenario, we repeat the same test case in Section \ref{sec:nonzerorelvel} with $Pe_c=1$, except that the total integration time is $t_{\mathrm{end}}=0.05$ and a larger domain of range $[0,5]$ is chosen with periodic boundary conditions. The drop is placed at the same location of $x_0=0.5$ with a non-uniform scalar concentration field given by the function with a compact support
\begin{equation}
c = \Bigg\{ 
    \begin{aligned}
        & A e^{\big\{\frac{1}{(4x - 2)^2 - 1}\big\}} & \hspace{10mm}  x\in [0.25,0.75] \\
        & 0 & \hspace{10mm} \mathrm{else},
    \end{aligned}
 \label{}   
\end{equation}
as shown in Figure \ref{fig:non_uniform_case}(a). Here, $A$ is a factor chosen to make the quantities $\int_{\Omega} c\ dV$ and $\int_{\Omega} \phi\ dV$ discretely equal, such that as the scalar diffuses and reaches a steady state within the drop, we should expect to see the scalar concentration values reach a uniform value of $1$ within the drop. Initial conditions for $\phi$ and $c$, the time evolution of $c$, and the final state are shown in Figure \ref{fig:non_uniform_case}(a,b,c), respectively. 
\begin{figure}
\vskip 0.1in
\centering
\includegraphics[width=\textwidth]{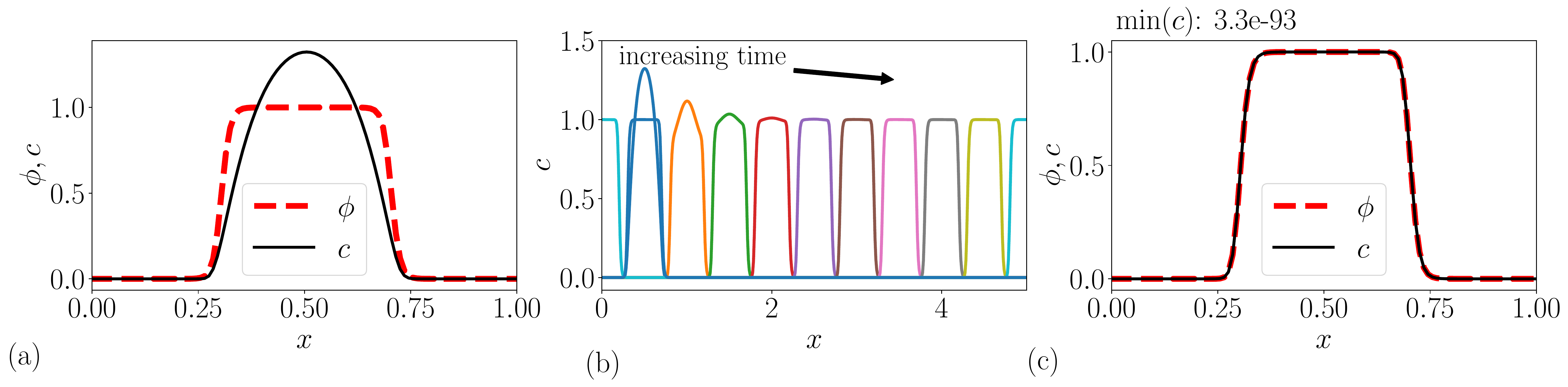}
\caption{The advection of a drop along with a non-uniformly distributed scalar quantity inside the drop. (a) The initial conditions for $\phi$ and $c$. (b) The time evolution of $c$. (c) The final states of $\phi$ and $c$.}
\label{fig:non_uniform_case}
\end{figure}
The scalar undergoes diffusion within the drop without any numerical leakage or encountering negative values, and at the final time, the scalar concentration is uniform within the drop, as expected, which illustrates the robustness of the proposed model.



\subsection{Bubble in a concentration polarization region \label{sec:scalar_bubble}}

In this section, a two-dimensional test case of transport of a passive scalar around the bubble in a channel is presented. A schematic of the domain is shown in Figure \ref{fig:scalar_around_bubble}. The domain is a square channel of size $L\times L$, where $L=0.1$. The streamwise direction has periodic boundary conditions and the other direction has no-slip walls on both ends. A gas bubble of radius $R=d/2=0.02$ is initially placed at the center of the channel, surrounded by a liquid. Because of the periodic boundary conditions, the single bubble in a channel setup essentially represents a train of equally-spaced bubbles in a channel. 

\begin{figure}
    \centering
    \includegraphics[width=0.4\textwidth]{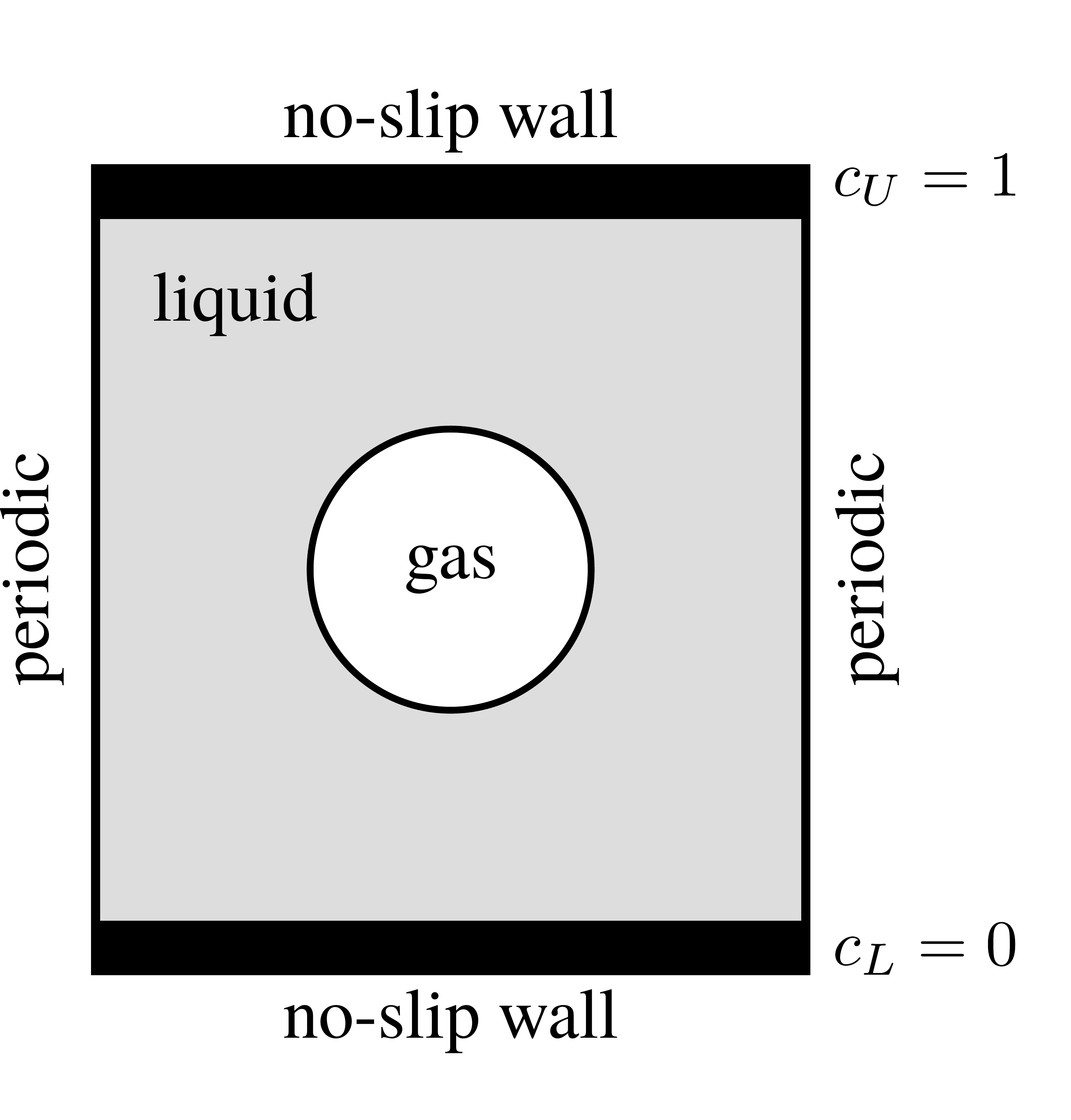}
    \caption{A schematic of the domain for the case of scalar diffusion around a bubble in a channel. Here, $c_U$ and $c_L$ represents the boundary conditions for the scalar concentration on the upper and lower walls, respectively.}
    \label{fig:scalar_around_bubble}
\end{figure}

Here, the scalar is confined to the liquid region. If the scalar quantity represents the temperature field, then this corresponds to the physical scenario of $Pr_l/Pr_g \rightarrow 0$, where $Pr$ is the Prandtl number, and the subscripts $l$ and $g$ represents liquid and gas fluid quantities. If the scalar instead represents dissolved salts or dye, then this corresponds to the scenario of $Sc_l/Sc_g \rightarrow 0$, where $Sc$ is the Schmidt number. In either case, the ratio of diffusivities is $D_g/D_l \rightarrow 0$, which essentially leads to the confinement of the scalar in the liquid region. 

In the current setup, the diffusivities are chosen to be $D_l=0.01$ and $D_g=0$. The scalar concentration is initially zero everywhere in the domain. The upper and lower walls have Dirichlet boundary conditions of $c_U=1$ and $c_L=0$, respectively. Without the presence of the bubble in the channel, the scalar concentration will remain uniform along the streamwise direction. The final steady state scalar concentration profile along the wall-normal direction can be obtained by solving a one-dimensional steady state heat equation, which would give a linear profile with a slope of $(c_U-c_L)/L=10$. When the bubble is present, it experiences the varying scalar concentration values around it which makes it a good test case to evaluate the proposed model. Moreover, this test case was chosen because the setup represents the concentration polarization region and the gas bubbles that are formed at the electrode in a electrochemical system. 

In the subsequent Sections \ref{sec:stat_bubble} and \ref{sec:mov_bubble}, two situations that correspond to (a) a stationary bubble, and (b) a moving bubble, have been studied. The importance of the proposed scalar-transport model in Eq. \ref{equ:interface} and the positivity criterion in Eq. \ref{equ:pecletcriterion} is also highlighted by comparing the results obtained using this model against the results obtained by solving the Eq. \ref{equ:c_mod}. Hereafter, we denote the solution obtained from Eq. \ref{equ:c_mod} as the ``previous approach".

\subsubsection{Stationary bubble\label{sec:stat_bubble}}

Here, the bubble is assumed to be stationary, and the fluid velocity is zero everywhere in the domain. This setup, therefore, requires solving only the scalar-transport equation. The time scale of diffusion of the scalar in the domain is $\tau_d\sim L^2/D = 1$, and we expect the scalar concentration to reach a steady state after this time. Therefore, the total time of the simulation is taken to be $t_{end}=2\tau_d=2$. 

%

\begin{figure}
    \centering
    \includegraphics[width=0.7\textwidth]{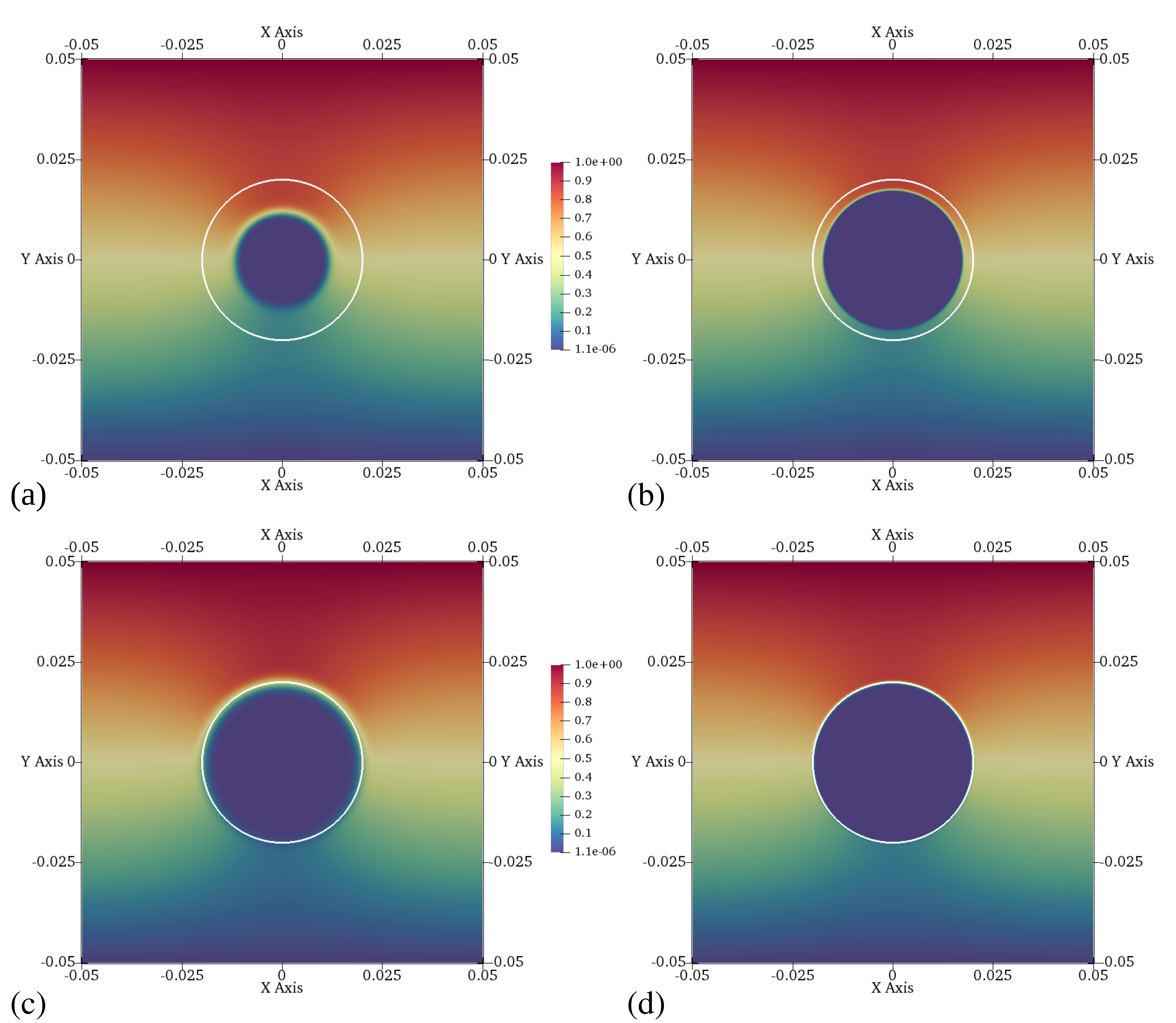}
    \caption{The final time steady state scalar concentration field obtained using, (a) the previous approach on a grid of size $128\times128$, (b) the previous approach on a grid of size $512\times512$, (c) the proposed model on a grid of size $128\times128$, and (d) the proposal model on a grid of size $512\times512$. The white solid line represents the interface; and the color plot represents the scalar concentration field.}
    \label{fig:stat_bubble}
\end{figure}

The final steady state scalar concentration fields obtained by solving the proposed model and from the previous approach are shown in Figure \ref{fig:stat_bubble}, on two different grids of sizes, $128\times128$ and $512\times512$. The previous approach results in the leakage of the scalar into the bubble for both grids; and increasing the grid resolution resulted in less leakage of the scalar. However, the results from the proposed model show no signs of leakage of the scalar into the bubble on either grids. The leakage of the scalar in the previous approach is more evident in Figure \ref{fig:stat_bubble_line}, where $c$ and $\phi$ are plotted along the wall-normal direction at $x=0$. 

\begin{figure}
    \centering
    \includegraphics[width=\textwidth]{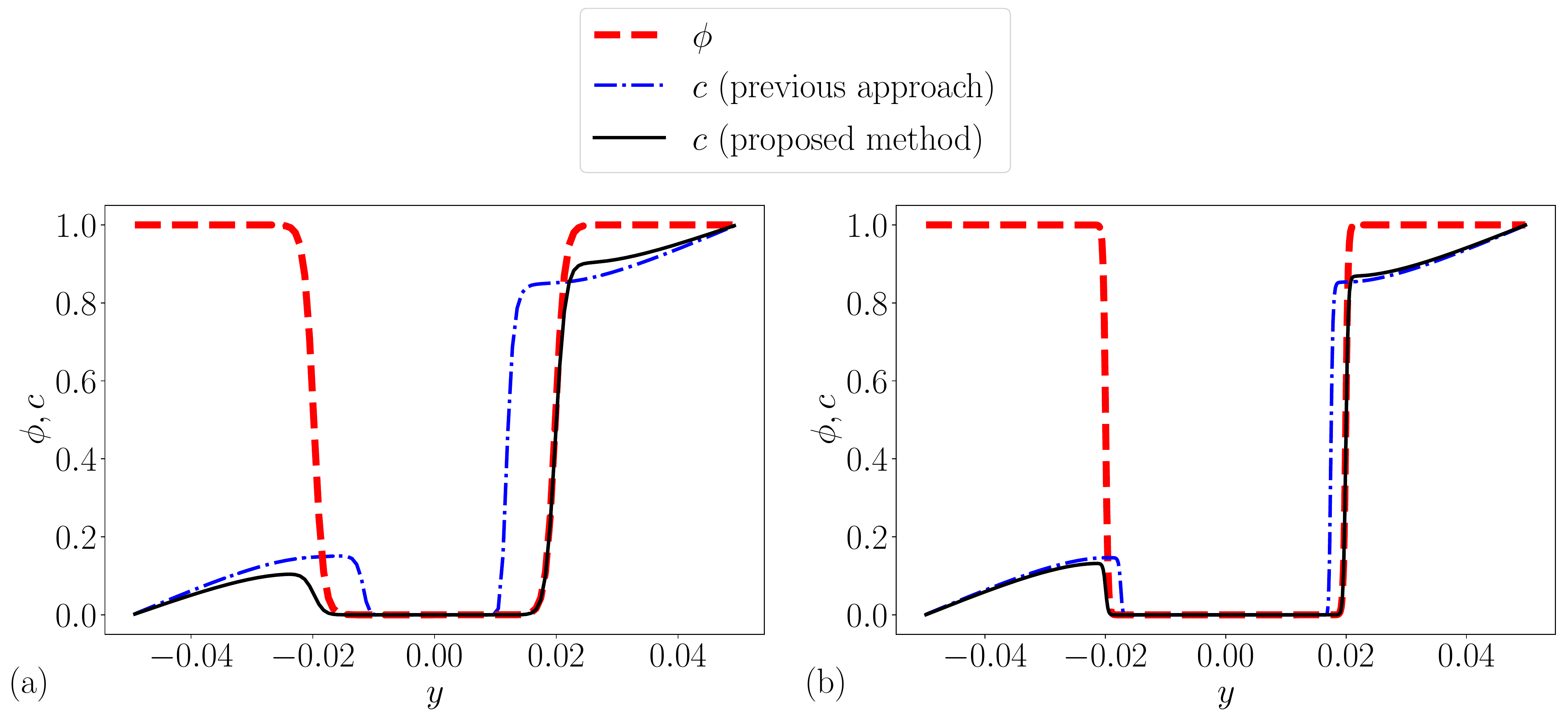}
    \caption{The scalar concentration field $c$ and the volume fraction field $\phi$ at the final time of $t=2$ along the $x=0$ line for the case of scalar diffusion around a stationary bubble in a channel, simulated on the grids of size (a) $128\times128$, and (b) $512\times512$.}
    \label{fig:stat_bubble_line}
\end{figure}

To further quantify the leakage seen in Figure \ref{fig:stat_bubble_line}, a quantity $c_e$ defined as 
\begin{equation}
c_e = \Bigg\{ 
    \begin{aligned}
        & |c-\phi| & \hspace{10mm}  \phi < 10^{-3} \\
        & 0 & \hspace{10mm} \mathrm{else},
    \end{aligned}
 \label{equ:leakage_error}   
\end{equation}
is computed along the wall-normal direction at $x=0$ and is plotted in Figure \ref{fig:stat_bubble_line_error}. Since $c_e$ is the difference between $c$ and $\phi$ within the bubble region, it approximately represents the effective amount of scalar that leaked into the bubble region. An exact method would result in $c_e$ being zero, which can be seen for the proposed method in Figure \ref{fig:stat_bubble_line_error}. Whereas, the previous approach results in the leakage of the large amounts of scalar into the bubble region.  
\begin{figure}
    \centering
    \includegraphics[width=\textwidth]{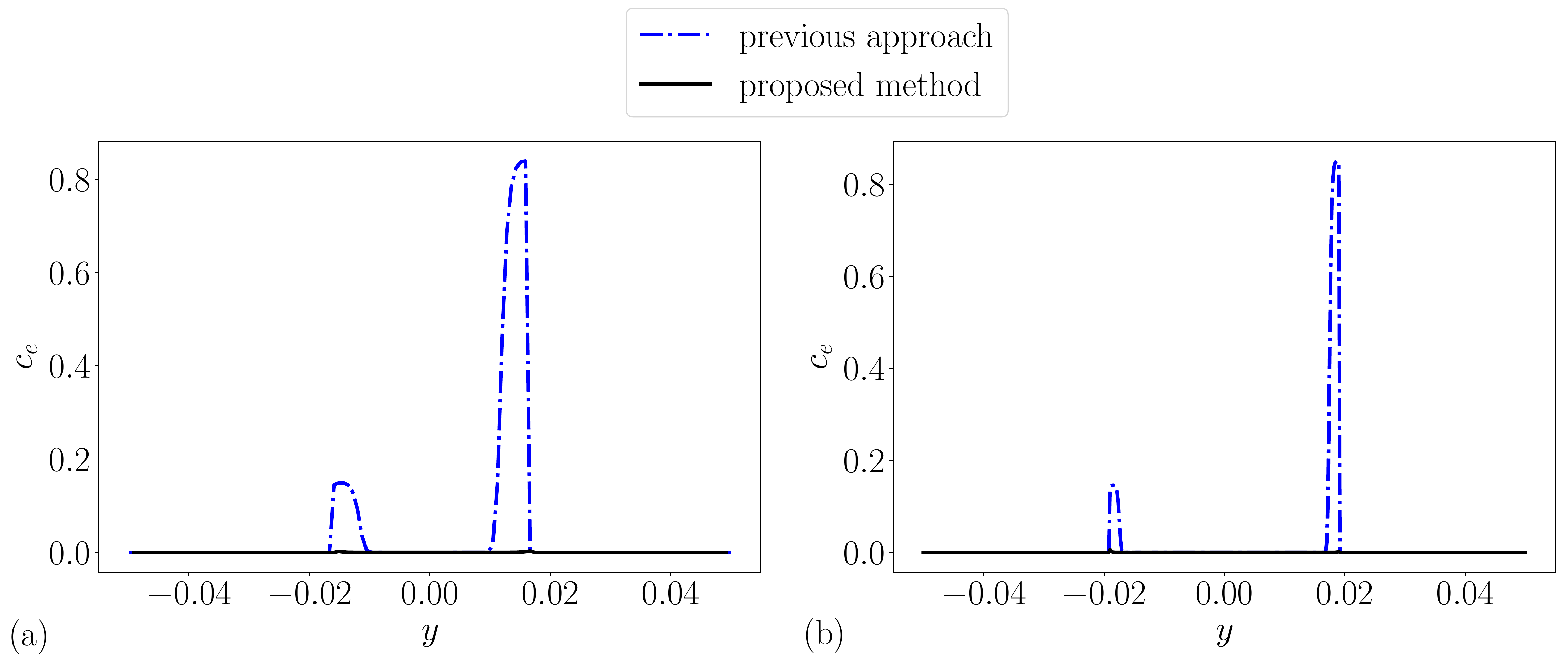}
    \caption{The scalar leakage error $c_e$ at the final time of $t=2$ along the $x=0$ line for the stationary-bubble case, simulated on the grids of size (a) $128\times128$, and (b) $512\times512$.}
    \label{fig:stat_bubble_line_error}
\end{figure}

Integrating $c_e$ along a line, an integrated leakage error metric $I_{error}$ can be defined as
\begin{equation}
    I_{error}=\int_s c_e\ dx dy.
    \label{equ:int_leakage_error}
\end{equation}
Values of $I_{error}$ obtained by integrating $c_e$ in Figure \ref{fig:stat_bubble_line_error} are listed in Table \ref{tab:stat_bubble_error}. The integrated leakage error with the proposed method is four orders of magnitude lower compared to the previous approach.

\begin{table}[]
\centering
\begin{tabular}{@{}|c|c|c|@{}}
\toprule
\textbf{$I_{error}$} & $128\times128$ & $512\times512$ \\ \midrule
previous approach & $5.96\times10^{-3}$ & $6.31\times10^{-3}$ \\
proposed method & $8.78\times10^{-7}$ & $2.34\times10^{-7}$ \\ \bottomrule
\end{tabular}
\caption{Integrated leakage error for the stationary-bubble case.}
\label{tab:stat_bubble_error}
\end{table}

\subsubsection{Moving bubble\label{sec:mov_bubble}}

Here, in this section, two situations of moving bubble in a channel are considered. In the first case, the bubble is assumed to be translating at a fixed velocity without undergoing deformation. This setup, therefore, requires solving the scalar-transport equation and the phase-field equation. Later, in the second case, the bubble and the surrounding liquid is initialized with the same velocity, but in this case the bubble is allowed to deform in an accelerating flow. Therefore, this requires solving the proposed model coupled with the hydrodynamics.

\begin{figure}
    \centering
    \includegraphics[width=0.8\textwidth]{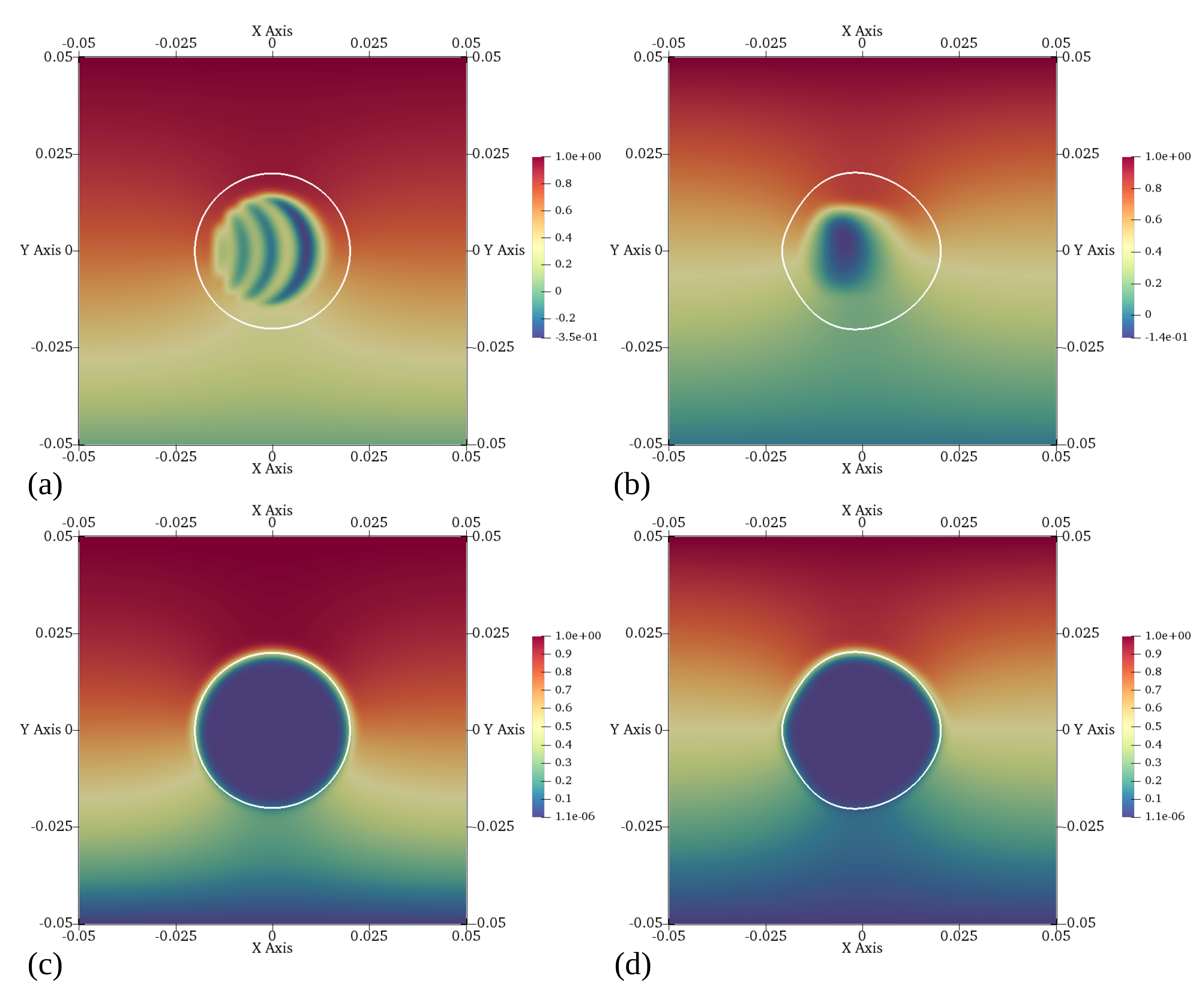}
    \caption{The scalar concentration field for the previous approach at (a) $t=0.1$, (b) $t=2$, and for the proposed model at (c) $t=0.1$, (d) $t=2$. The white solid line represents the interface; and the color plot represents the scalar concentration field.}
    \label{fig:moving_bubble}
\end{figure}

For the first case of a translating bubble, the fluid velocity is $\vec{u}=1$ everywhere in the domain. This case is more challenging compared to the stationary bubble case in Section \ref{sec:stat_bubble}, because of the non-zero droplet-Peclet number ($Pe = ud/D = 4$) of this flow. The scalar concentration fields obtained by solving the proposed model and from the previous approach are shown in Figure \ref{fig:moving_bubble} at two different times, $t=0.1$, and $t=2$, on a grid of size $128\times128$. The results from the previous approach not only show the leakage of scalar into the bubble, but also exhibit the characteristic dispersion errors associated with a non-dissipative scheme, at early times ($t=0.1$). At the final time of $t=2$, the scalar almost completely leaks into the bubble. The dispersion errors generate negative values of the scalar concentration and results in realizability issues of the scalar concentration field. The negative values in the results from the previous approach is evident in Figure \ref{fig:moving_bubble_line}, where $c$ and $\phi$ are plotted along the wall-normal direction at $x=0$. Unlike the results from the previous approach, the proposed model neither show signs of leakage of the scalar into the bubble, nor the dispersion errors. 

\begin{figure}
    \centering
    \includegraphics[width=\textwidth]{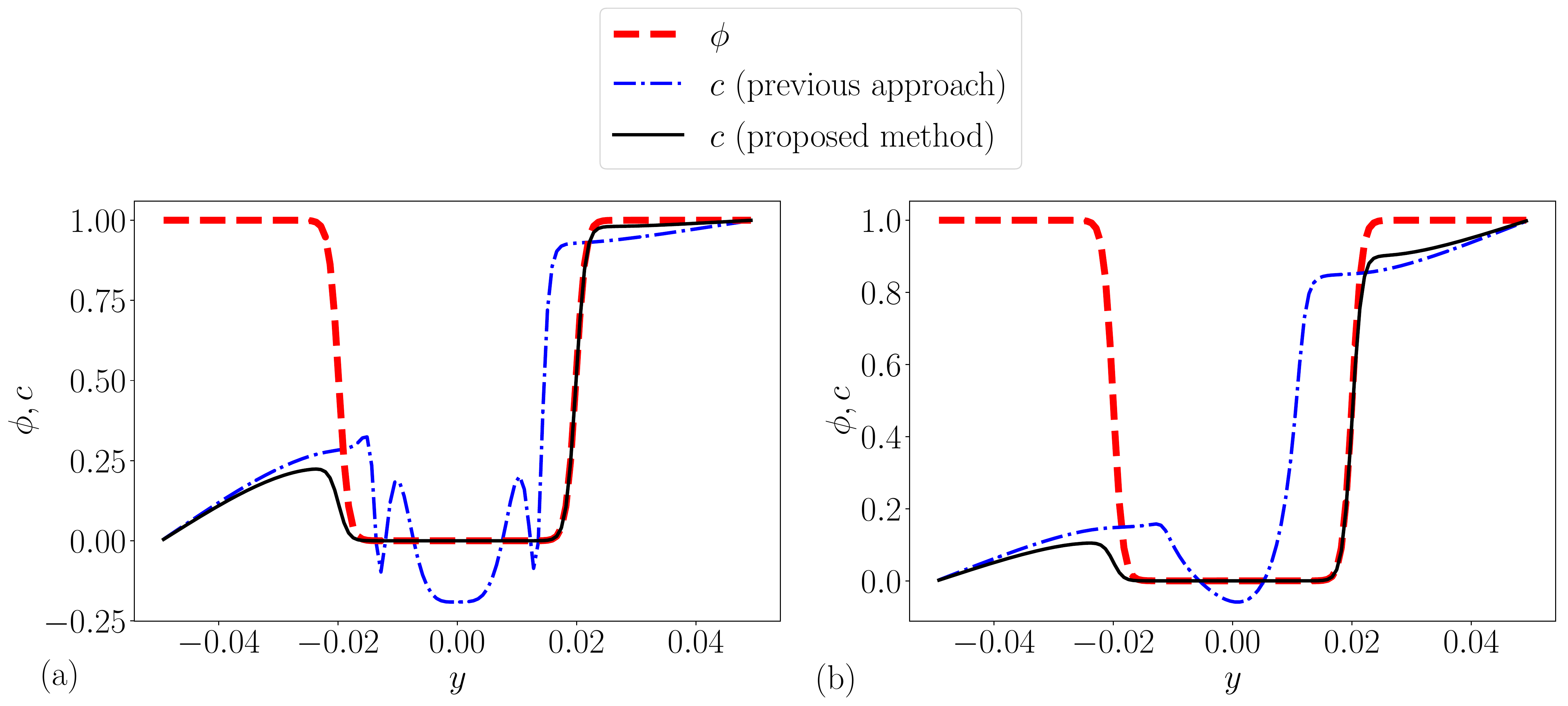}
    \caption{The scalar concentration field $c$ and the volume fraction field $\phi$ along the $x=0$ line for the case of scalar diffusion around a moving bubble in a channel, simulated on a grid of size $128\times128$, at (a) $t=0.1$, and (b) $t=2$.}
    \label{fig:moving_bubble_line}
\end{figure}

The scalar leakage error $c_e$ defined in Eq. \eqref{equ:leakage_error} and computed along the wall-normal direction at $x=0$ is plotted in Figure \ref{fig:moving_bubble_line_error}. The non-zero values of $c_e$ for the previous approach imply a significant leakage of the scalar into the bubble region. The integrated leakage error defined in Eq. \eqref{equ:int_leakage_error} and evaluated along the wall-normal direction at $x=0$ is listed in Table \ref{tab:moving_bubble_error}. Similar to the stationary-bubble case, the integrated leakage error with the proposed method is significantly lower compared to the previous approach. 

\begin{figure}
    \centering
    \includegraphics[width=\textwidth]{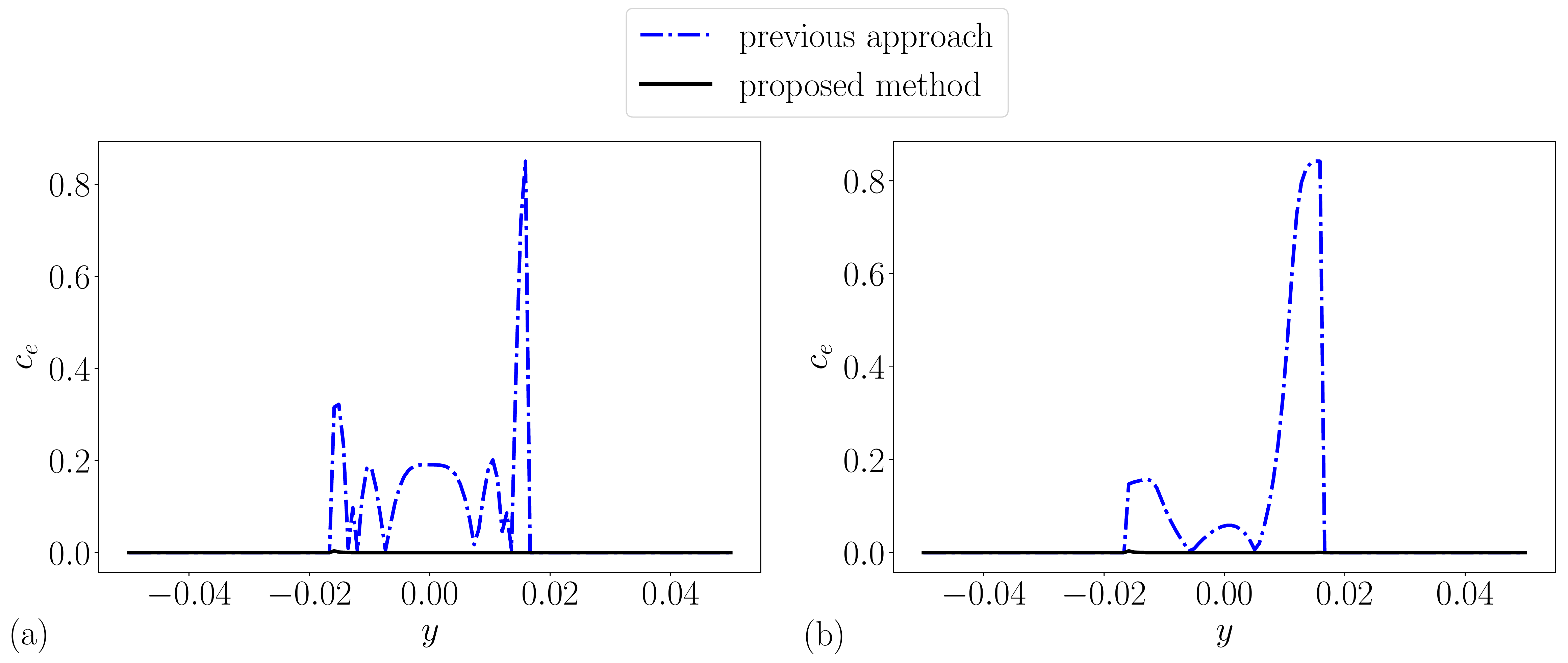}
    \caption{The scalar leakage error $c_e$ along the $x=0$ line for the moving-bubble case, simulated on a grid of size $128\times128$, at (a) $t=0.1$, and (b) $t=2$.}
    \label{fig:moving_bubble_line_error}
\end{figure}

\begin{table}[]
\centering
\begin{tabular}{@{}|c|c|c|@{}}
\toprule
\textbf{$I_{error}$} & $t=0.1$ & $t=2$ \\ \midrule
previous approach & $3.54\times10^{-2}$ & $3.79\times10^{-2}$ \\
proposed method & $6.27\times10^{-7}$ & $6.28\times10^{-7}$ \\ \bottomrule
\end{tabular}
\caption{Integrated leakage error for the moving-bubble case.}
\label{tab:moving_bubble_error}
\end{table}

In the second case, for a deforming bubble, the initial fluid velocity is taken to be $\vec{u}=1$ everywhere in the domain. The flow is sustained with a unit acceleration, $g=1$ along the streamwise direction. The density of the gas and the liquid are $\rho_g=1$ and $\rho_l=10$, respectively. The viscosity of the gas and the liquid are $\mu_g=1.81\times10^{-5}$ and $\mu_l=8.9\times10^{-4}$, respectively. This corresponds to a liquid Reynolds number of about $Re_l=\rho_l u d/\mu_l=450$. The surface tension is varied from $\sigma=0$ to $\sigma=5$, to study the effect of bubble breakup on the evolution of the scalar concentration field. This corresponds to the E{\"o}tv{\"o}s numbers of $Eo=\infty$ to $Eo=3.2\times10^{-3}$, respectively, where the E{\"o}tv{\"o}s number is defined as $Eo=\Delta \rho g d^2/\sigma$. 

\begin{figure}
    \centering
    \includegraphics[width=\textwidth]{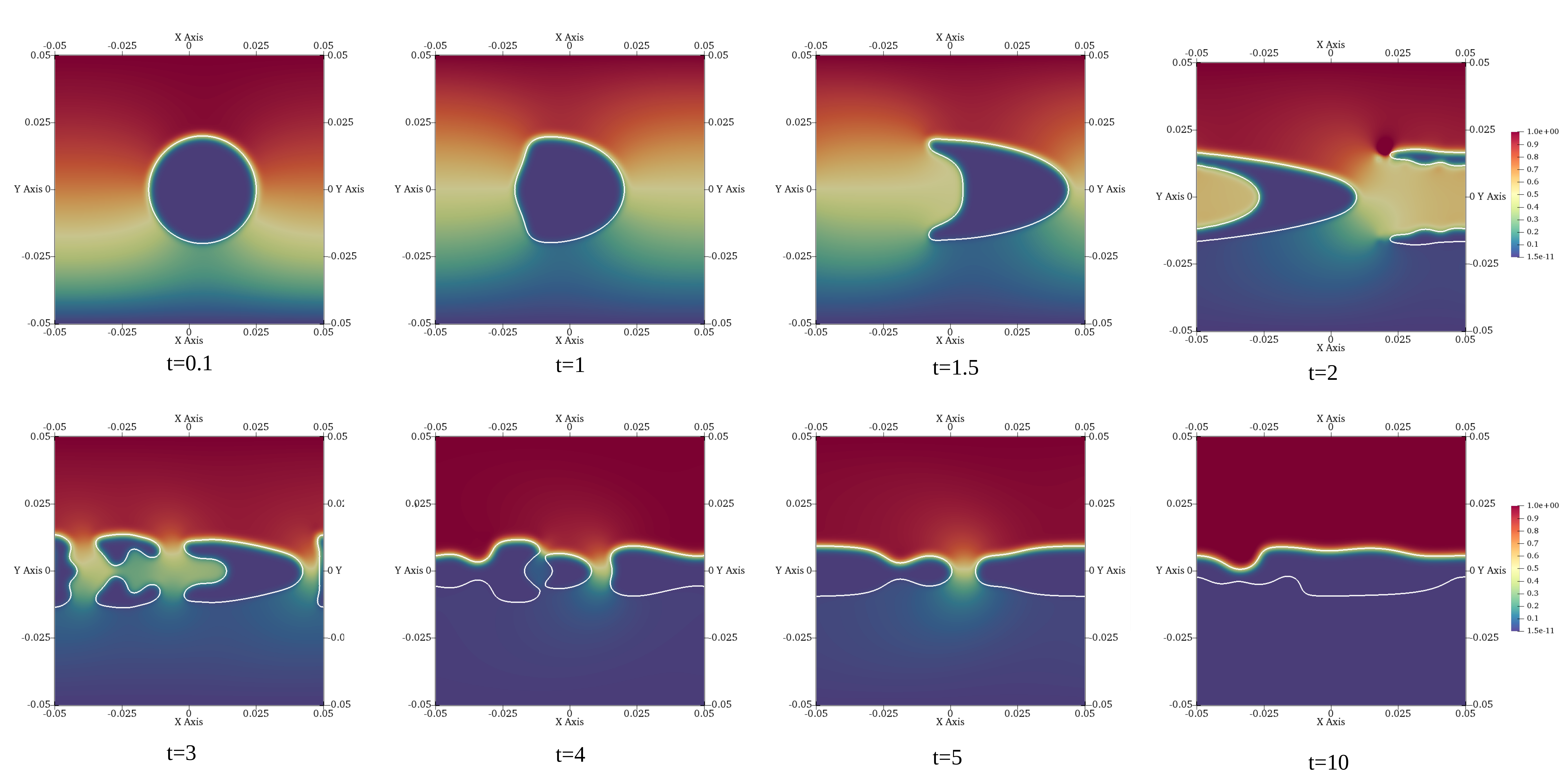}
    \caption{The evolution of the scalar concentration field around a deforming bubble in a channel for a surface tension value of $\sigma=0$. The white solid line represents the interface; and the color plot represents the scalar concentration field.}
    \label{fig:sigma0}
\end{figure}

Figure \ref{fig:sigma0} shows the evolution of the bubble shape and the scalar concentration field at various times up to $t=10$ for the surface tension value $\sigma=0$ ($Eo=\infty$), simulated on a grid of size $128\times128$. The bubble deforms due to the relative difference in the body force between the bubble and the surrounding liquid, which is effectively a buoyancy force, and develops skirt on the two ends as can be seen at $t=1.5$ and $t=2$. \citet{clift2005bubbles} characterized the shapes of the bubble on a $Re-Eo$ plot, and the $Re=450$ and $Eo=\infty$ corresponds to a skirted spherical-cap regime, which matches well with our observation, although the current setup is a two dimensional one. Later, the bubble breaks up leaving behind a trail of small bubbles, which eventually fills the channel as can be seen at $t=10$. With an increase in the surface tension values (decrease in $Eo$) the deformation of the bubble is reduced. Figure \ref{fig:sigma0_0005} shows the time evolution of bubble shape and the scalar concentration for the surface tension value of $\sigma=5\times10^{-4}$ ($Eo=32$). 

\begin{figure}
    \centering
    \includegraphics[width=\textwidth]{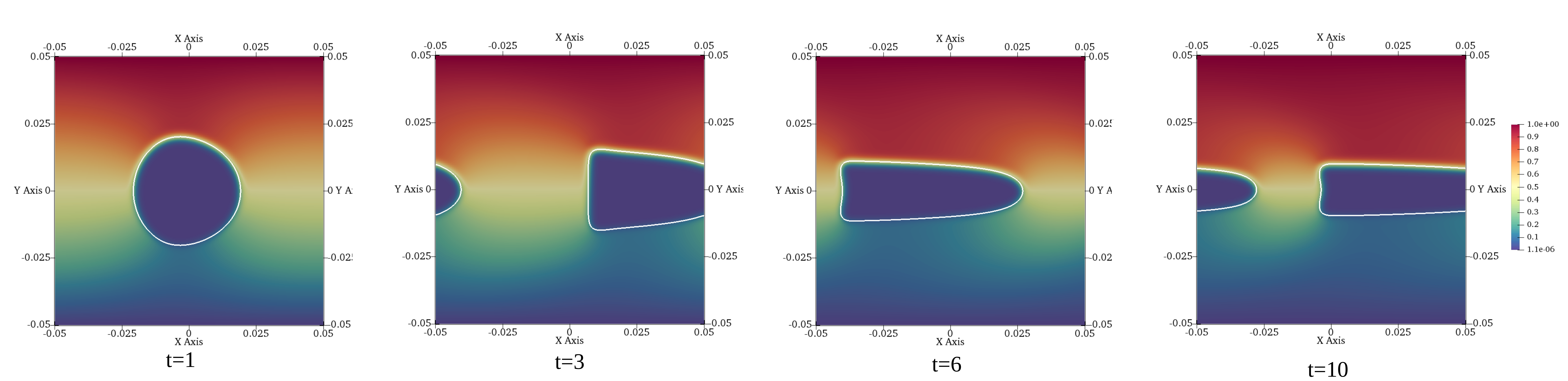}
    \caption{The evolution of the scalar concentration field around a deforming bubble in a channel for a surface tension value of $\sigma=5\times10^{-4}$. The white solid line represents the interface; and the color plot represents the scalar concentration field.}
    \label{fig:sigma0_0005}
\end{figure}

One consequence of the bubble deformation and breakup is that the flux of the scalar across the channel in the wall-normal direction is reduced because of the reduced cross-sectional area. This can be quantified by calculating the flux of the scalar at the wall defined as $f_w=D\lvert \vec{\nabla} c \rvert$. The average flux of the scalar, $\langle f_w \rangle$, is plotted in Figure \ref{fig:flux} versus time for various values of surface tension, where the average is computed over top and bottom walls. For the case of $\sigma=0$, it was seen in Figure \ref{fig:sigma0} that the bubble would eventually span the whole channel blocking the flow of scalar. As a result, the flux of the scalar can be seen to drop to zero for time $t\gtrapprox8$. With an increase in the surface tension values, the deformation of the bubble decreases, and therefore the flux of the scalar increases. At the surface tension value of $\sigma=5$, the bubble essentially remains circular throughout simulation; and the average scalar flux reaches a limiting value of $\langle f_w \rangle = 0.075445$ that was obtained for the stationary bubble case in Section \ref{sec:stat_bubble}.  

\begin{figure}
    \centering
    \includegraphics[width=0.8\textwidth]{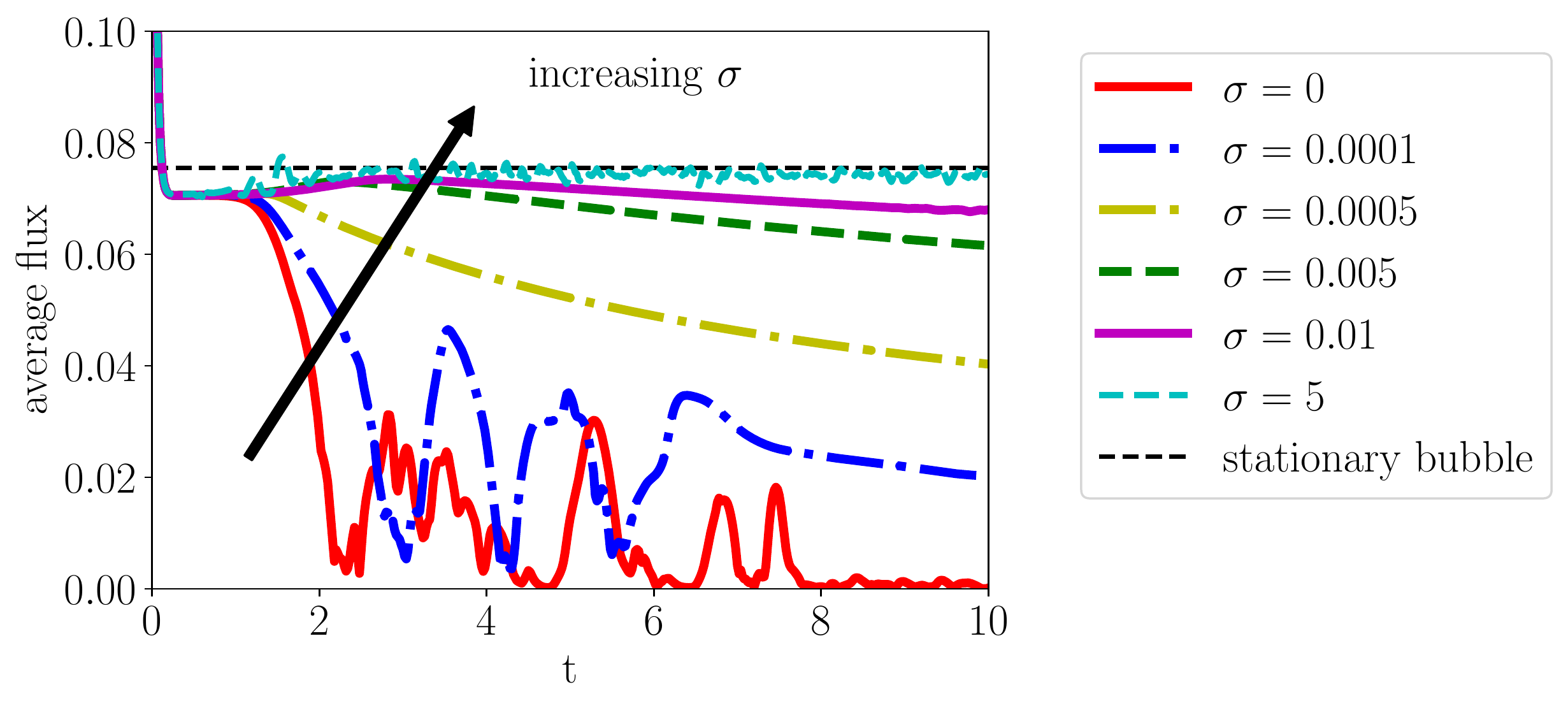}
    \caption{The average scalar flux versus time for various values of surface tension. The thin black dashed line represents the average flux for the stationary bubble case in Section \ref{sec:stat_bubble}.}
    \label{fig:flux}
\end{figure}

\subsection{Charged ions in a drop \label{sec:ion_drop}}

In this section, the simulation of reorganization of ions with unbalanced charge within a drop will be presented. This case illustrates the applicability of the proposed scalar-transport model in Eq. \ref{equ:interface} for modeling electrokinetics in two-phase flows. To model this phenomenon, the scalar-transport equation is recast into the Nernst-Planck equation
\begin{equation}
    \frac{\partial c^{\pm}}{\partial t} + \vec{\nabla}\cdot(\vec{u} c^{\pm} \pm \phi \lambda \vec{E} c^{\pm}) = \vec{\nabla} \cdot \left[D^{\pm} \left\{\vec{\nabla}c^{\pm} - \frac{(1 - \phi) \vec{n} c^{\pm}}{\epsilon} \right\}\right],
    \label{eq:nernstplanck}
\end{equation}
where $c^+$ and $c^-$ represent the cationic and anionic concentration fields, respectively; $D^+$ and $D^-$ represent the cationic and anionic diffusion coefficients, respectively; $\lambda$ is the electrical mobility and $\vec{E}$ is the electric field. Here, the relative velocity $\vec{u}_r$ is replaced by the effective electromigration velocity experienced by the ions, and is given by $\vec{u}_r=\lambda \vec{E}$. To close the system of equations, the Nernst-Planck equation is combined with the Gauss's law
\begin{equation}
    \vec{\nabla}\cdot(\varepsilon \vec{E}) = \rho^f,
    \label{eq:gauss}
\end{equation}
where $\rho^f=z^+ e c^+ - z^- e c^-$ is the free charge density; $z^+$ and $z^-$ represent the cationic and anionic valences, respectively; $e$ is the elementary charge; and $\varepsilon$ is the electrical permittivity of the electrolyte. Making an electrostatic approximation, i.e., assuming that the time variation of the magnetic field is much slower compared to that of the electric field, the electric field can be shown to be irrotational using Faraday's law \citep{griffiths2005introduction}. Therefore, the electric field can be written as
\begin{equation}
    \vec{E}=-\vec{\nabla} \psi,
    \label{eq:electrostatic}
\end{equation}
where $\psi$ is the electrostatic potential. Now, invoking the Einstein-Smoluchowski relation 
\begin{equation}
    D^{\pm} = \frac{\lambda k_B T}{z^{\pm} e},
\end{equation}
where $k_B$ is the Boltzmann constant, and by re-scaling the electrostatic potential $\psi$, with the thermal voltage $V_T$, as 
\begin{equation}
    \tilde{\psi} = \frac{\psi}{V_T},
\end{equation}
where $V_T = k_B T/(z^{\pm}e)$, and $T$ is the absolute temperature, the Nernst-Planck equation in Eq. \ref{eq:nernstplanck} can be rewritten in the form 
\begin{equation}
     \frac{\partial c^{\pm}}{\partial t} + \vec{\nabla}\cdot\big\{\vec{u} c^{\pm} \mp \phi D^{\pm} (\vec{\nabla}\tilde{\psi}) c^{\pm}\big\} = \vec{\nabla} \cdot \left[D^{\pm} \left\{\vec{\nabla}c^{\pm} - \frac{(1 - \phi) \vec{n} c^{\pm}}{\epsilon} \right\}\right],
\end{equation}
which will be used in the current setup. Electrostatic phenomena and hydrodynamics can be coupled through the Maxwell stress tensor either as a body force or as a divergence of the tensor \citep{saville1997electrohydrodynamics, berry2013multiphase}, and will be part of a separate study and is beyond the scope of this work.

In this study, the setup consists of a stationary circular liquid drop of radius $R=0.2$, in a quiescent surrounding medium of another fluid. The density of both fluids are chosen to be $\rho_d=\rho_s=1$, where the subscripts $d$ and $s$ represents the drop and surrounding fluid properties. The viscosity of both fluids are chosen to be $\mu_d=\mu_s=10^{-3}$. The ion concentrations are initially uniform within the drop and are normalized by the reference value $c_{ref}=10^9$ as $\tilde{c}^{\pm}=c^{\pm}/c_{ref}$. The normalized initial ion concentrations are chosen to be $\tilde{c}^+=1$, and $\tilde{c}^-=0$. The diffusivity of the ions are chosen to be $D^{\pm}_d=0.01$ and $D^{\pm}_s=0$. 

\begin{figure}
    \centering
    \includegraphics[width=\textwidth]{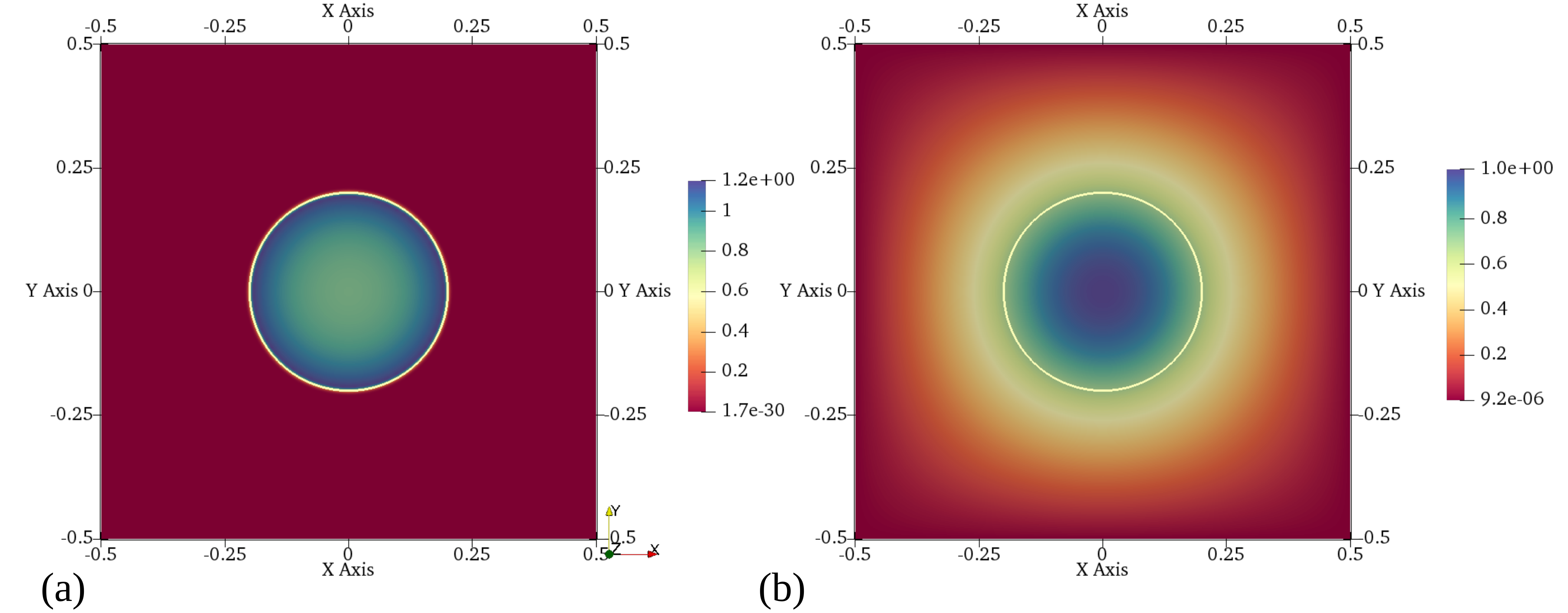}
    \caption{(a) The normalized ion concentration field, and (b) the re-scaled electrostatic potential field for the case of charged ions in a drop at time $t=1$, simulated on a grid of size $512\times512$.}
    \label{fig:charge_reorg}
\end{figure}

Figure \ref{fig:charge_reorg} shows the ion concentration field $\tilde{c}^+$ and the resulting electrostatic potential field $\tilde{\psi}$ at time $t=1$, simulated on a grid of size $512\times512$. The ions in the drop experiences the Columbic repulsive force because of the unbalanced net positive charge. This force drives the ions apart, towards the interface. This is counteracted by the diffusion in the opposite direction. The ions, however, cannot cross the interface because of the non-conducting surrounding medium; and therefore, they arrive at a state where the electromigration is balanced by the diffusion. Figure \ref{fig:charge_reorg_line} shows the ion concentration $\tilde{c}^+$ and the resulting electrostatic potential $\tilde{\psi}$ at time $t=1$ along the line $x=0$ in the domain. The results of the simulation from four different grid sizes were chosen to show the convergence of the results.   


\begin{figure}
    \centering
    \includegraphics[width=\textwidth]{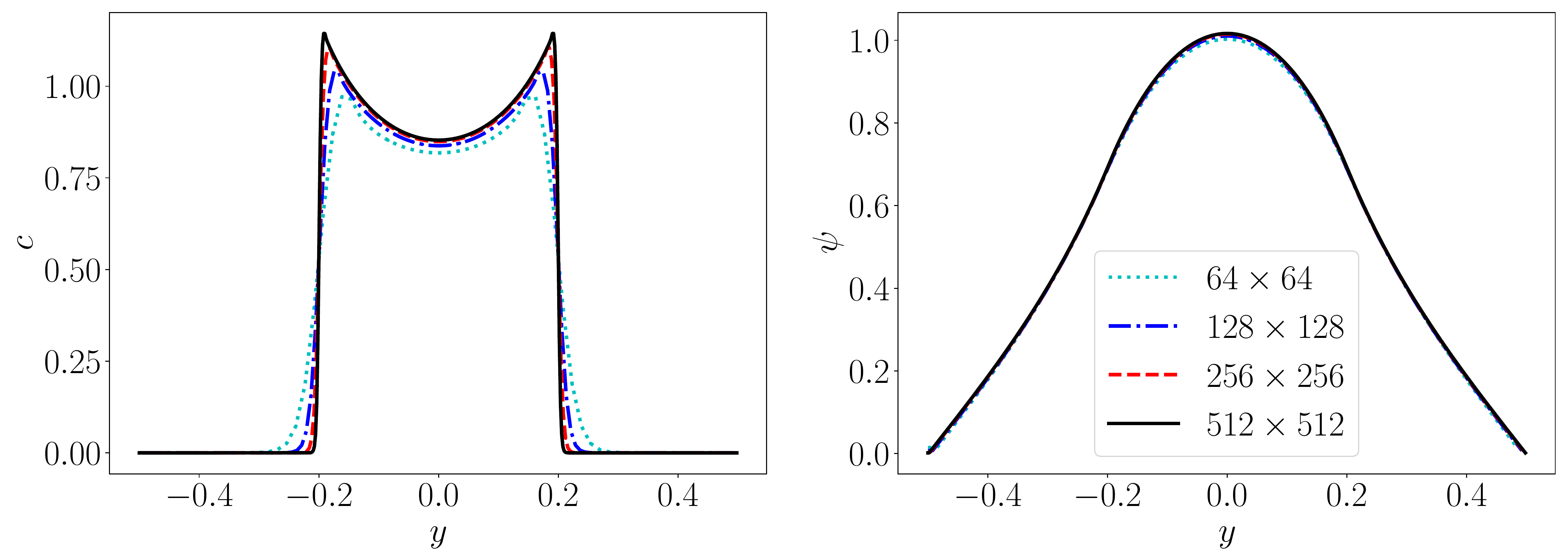}
    \caption{(a) The normalized ion concentration $\tilde{c}^+$, and (b) the re-scaled electrostatic potential $\tilde{\psi}$ along the $x=0$ line for the case of charged ions in a drop at time $t=1$, simulated on grids of different sizes.}
    \label{fig:charge_reorg_line}
\end{figure}

\subsection{Droplet-laden turbulent channel flow \label{sec:scalar_turbulent}}

Up until here in Section \ref{sec:results}, simple tests cases in 1D and 2D were presented. Here, a direct numerical simulation of a droplet-laden turbulent channel flow with passive scalar quantity that is confined to the surrounding fluid (carrier fluid) is presented. This test case illustrates: the validity of the positivity criterion, the applicability and robustness of the proposed scalar-transport model, and the prevention of unphysical numerical leakage by the model for complex two-phase turbulent flow regimes.  

The setup consists of a three-dimensional channel of size $8H\times2H\times4H$ in $x,y,z$ as shown in Figure \ref{fig:channel_domain}, where $H=1$ is the channel half height. The streamwise direction (along the $x$ coordinate) and the spanwise direction (along the $z$ coordinate) have periodic boundary conditions, and the wall-normal direction (along the $y$ coordinate) has no-slip walls. 

\begin{figure}
    \centering
    \includegraphics[width=0.7\textwidth]{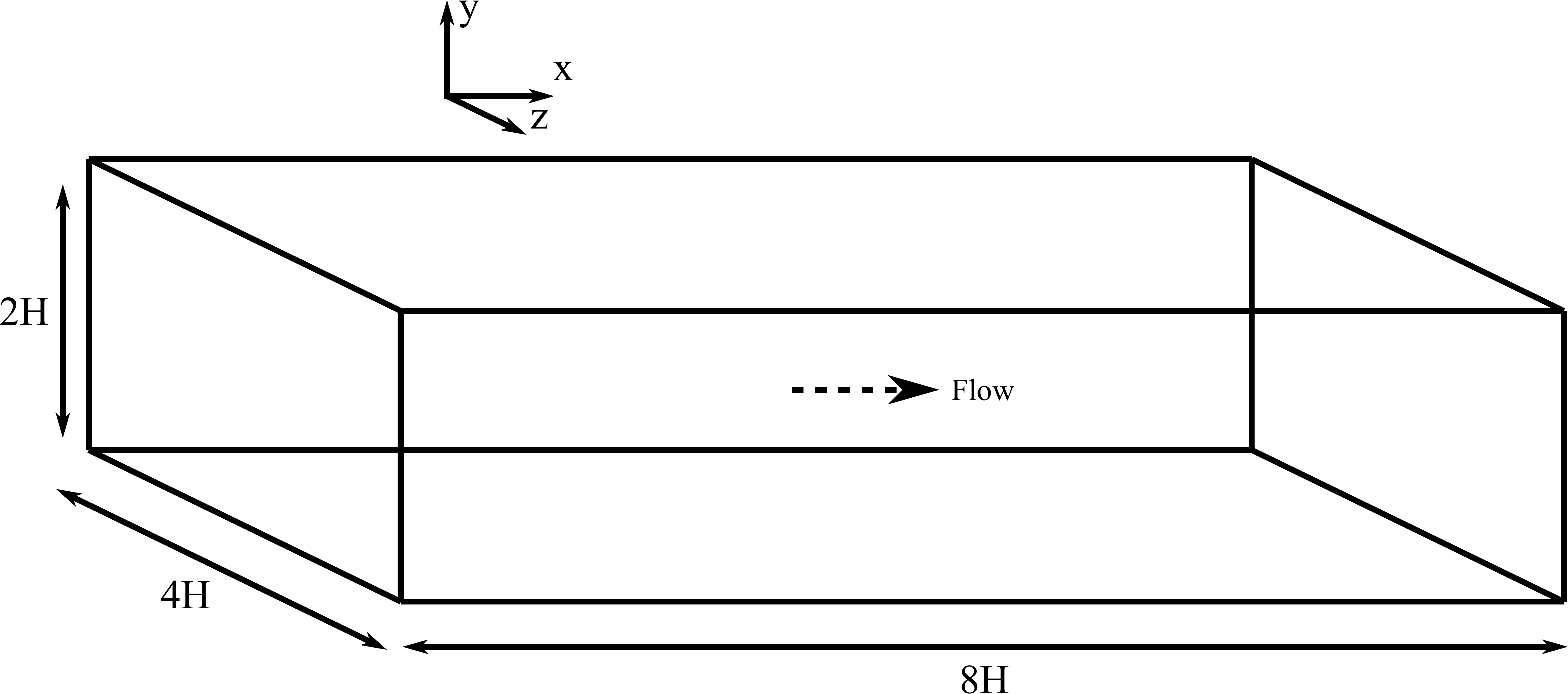}
    \caption{A schematic of the domain for the case of droplet-laden channel flow.}
    \label{fig:channel_domain}
\end{figure}

A unity density ratio $\rho*=\rho_d/\rho_s=1$ and a unity viscosity ratio $\mu^*=\mu_d/\mu_s=1$ are considered in this case such that the kinematic viscosities $\nu=\mu_d/\rho_d=\mu_s/\rho_s = 0.01$ are same for the droplet fluid and the surrounding fluid. The computation is performed for a Reynolds number of $Re=UH/\nu = 3000$, defined based on the mean centerline velocity $U$. This corresponds to a shear Reynolds number of $Re_{\tau}=u_{\tau}H/\nu = 150$, which is similar to the droplet-laden turbulent flow study by \citet{scarbolo2015coalescence}. 

The domain is discretized into a uniform grid of size $256\times64\times128$, which results in a grid spacing of size $\Delta y^w=4.6875$ in wall units, where the superscript $w$ represents a non-dimensional quantity scaled by the wall variables; e.g. $\Delta y^w=(\Delta y) u_{\tau}/\nu$, where $\nu$ is the kinematic viscosity, $u_{\tau}=\sqrt{\tau_w/\rho}$ is the shear velocity, and $\tau_w$ is the wall shear stress. No subgrid-scale model is used in the computation, since the grid resolution is considered to be sufficiently fine to resolve the essential turbulent scales \citep{kim1987turbulence}. 

A single-phase channel flow with surrounding fluid properties is simulated until it reaches a statistically stationary state. Then, droplets of diameter $d=0.4$ are added to the channel at a distance of $y^+=54$ from the walls, and are placed in the flow such that their centers are equally spaced in the $x-z$ plane. A total of $100$ droplets are added and the resulting volume fraction is $0.054$. A passive scalar quantity of diffusivities $D_s=0.5$ and $D_d=0$, in the surrounding fluid and the droplet fluid, respectively, is added to the surrounding fluid with a uniform initial concentration of $1$ at the same time when the droplets are introduced to the flow. The scalar is confined to the surrounding fluid region since the ratio of diffusivities in two fluids is $D_d/D_s=0$. The corresponding Schmidt numbers are $Sc_s=0.02$ and $Sc_d=\infty$, in the surrounding fluid and the droplet fluid, respectively. Since the $Sc_s$ is less than $1$, we expect that the scalar microscale, $\eta_s$, to be larger than the Kolmogorov scale, $\eta$, in the flow and it scales as $\eta_s\sim (D_s/\nu)^{3/4}\eta$ \citep{davidson2015turbulence}. Therefore, the grid requirement for the surrounding flow would be sufficient to resolve the smallest scales of the scalar. However, if $Sc_s$ was larger than $1$, then a more refined grid for the scalar needs to be used to resolve the Batchelor scale \citep[see,][]{schwertfirm2007dns}.  

\begin{figure}
    \centering
    \includegraphics[width=0.9\textwidth]{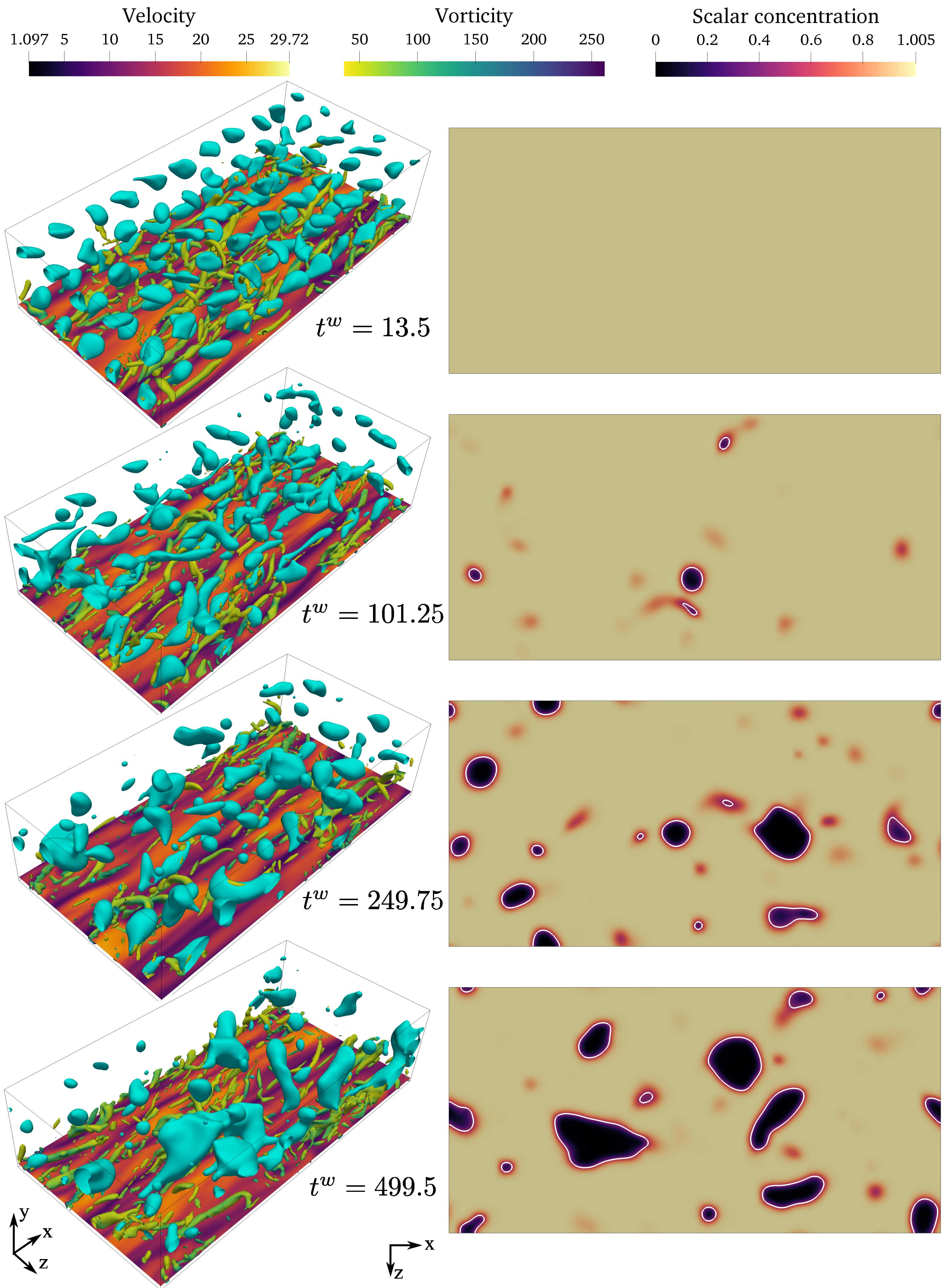}
    \caption{The snapshots of the droplet-laden turbulent channel flow at various times. The left column shows the three-dimensional view of the droplets and the flow structures in the channel. The blue colored surfaces represent the droplet interface. The flow structures are shown on the bottom half of the channel by plotting the isosurfaces of the second-invariant of the velocity gradient tensor, $Q=500$, colored by the local vorticity magnitude in the flow. The velocity magnitude is also plotted on a $x-z$ plane at a distance of $y=0.2H$ from the bottom wall. The right column shows the scalar concentration field on a $x-z$ cross-section midplane in the channel. The solid white lines represent the droplet interface.}
    \label{fig:turb_channel}
\end{figure}

The surface tension between the droplet fluid and the surrounding fluid is chosen such that the droplet-shear Weber number is $We_{\tau,d}=\rho u_{\tau}^2 d/\sigma=1$. Since the droplets are of size $d^w=96$ in wall units, they are larger than the Kolmogorov scale of $\eta^w\approx2$ wall units; and they undergo breakup and coalescence. The snapshots of the flow are shown in Figure \ref{fig:turb_channel} at various times. It appears from the snapshots that the droplet size is increasing, and therefore the number of droplets is decreasing with time, which is consistent with the observation of \citet{scarbolo2015coalescence}. For more details on the dynamics of the breakup and coalescence of the droplets, refer to the study by \citet{scarbolo2015coalescence}. Figure \ref{fig:turb_channel} also shows the scalar concentration field along the $x-z$ cross-section center plane of the channel at various times. Since the droplets were seeded at a distance of $y^+=54$ from the walls, the cross-section plane do not cut any of the droplets at the early time of $t^w=13.5$, therefore no droplets can be seen at this time. However, at later times, the droplets undergo coalescence and breakup and therefore can be seen in this view. Evidently, there was no leakage of the scalar quantity into the droplets or violation of the positivity of the scalar throughout the simulation. This illustrates the robustness of the proposed method and the validity of the positivity criterion. 

In this case, the maximum cell-Peclet number of the simulation is around $Pe_c=1.875$; and the positivity of the scalar is maintained throughout the simulation thought it violates the Positivity criterion by a small amount. This is because the Positivity criterion is a sufficient condition, and therefore cell-Peclet number values higher than $1$ might also result in maintaining positivity for the scalar quantity, as already illustrated in Section \ref{sec:nonzerorelvel} for a 1D case. However, repeating the droplet-laden channel flow simulation with a coarser grid (half the size of the original grid size in every direction) with a maximum cell-Peclet number of $Pe_c=3.75$, resulted in violation of the positivity of the scalar. Therefore, the grid size needs to selected such that the cell-Peclet number is maintained close to or less than $1$.

\section{Conclusions\label{sec:conclusion}} 

In this work, we proposed a novel transport model for the simulation of scalars in two-phase flows. The scalars are usually confined to one of the phases in a two-phase flow due to its disparate values of diffusivity and mobility in the two phases; and this typically poses a challenge for any numerical method in resolving the gradient of the scalar at the material interface.  

We therefore developed and verified a general scalar-transport model for two-phase flows, particularly for interfaces modeled using a phase-field (diffuse-interface) method. We showed that our newly proposed model equation prevents the artificial numerical diffusion (unphysical leakage) of the scalar from one phase to the other, while maintaining the positive values for the scalar concentration field throughout the simulation, albeit the use of central-difference schemes for the discretization of the operators. The use of central-difference scheme is to achieve a non-dissipative implementation that is crucial for the simulation of turbulent flows. 

We proved that the model maintains the positivity of the scalar concentration values\textemdash  a crucial realizability requirement for the simulation of scalars\textemdash provided the grid resolution was fine enough to satisfy the given positivity criterion. The grid resolution required to satisfy the positivity criterion is also based on the fact that the grid size should be small enough to resolve the smallest physical scales present in the flow. 

The prevention of unphysical numerical leakage of the scalar across the interface was achieved by enforcing consistency between the transport of the scalar concentration field and the phase field. It was shown that the equilibrium solution for the proposed scalar-transport model also exhibits a hyperbolic tangent function similar to the equilibrium solution for the phase-field equation, thereby making the transport of the scalar and the phase field consistent.

At the end, the proposed model was assessed of its accuracy, robustness, effectiveness in maintaining the positivity, and in preventing unphysical leakage across the interface, by simulating a wide range of two-phase flows, starting from simple one-dimensional droplet advection flows to complex three-dimensional droplet-laden turbulent flows. The droplet advection flows were used to verify the positivity criterion by choosing parameters that both violate and satisfy the positivity criterion, and by showing that the positivity is maintained in all the cases that satisfied the criterion. The proposed model was also recast into a Nernst-Planck equation and was used to simulate electrokinetics of two-phase flows. Finally, the droplet-laden turbulent channel flow showed the robustness of the model in simulating scalars in complex high-Reynolds-number turbulent two-phase flows.


\section*{Acknowledgments} 

This investigation was supported by the Office of Naval Research, Grant \# N00014-19-1-2425. S. S. Jain is also funded by the Franklin P. and Caroline M. Johnson Fellowship. The computations in this paper were performed on the Shepard, Armstrong and Yellowstone clusters at the Stanford HPC Center, supported through awards from the National Science Foundation, DOD HPCMP, and Office of Naval Research. A preliminary report on this work has been published in \citet{Jain2019scalar}  as a Center-for-Turbulence-Research Annual Research Brief and the authors acknowledge Dr. Ronald Chan's helpful comments on this report.

%
%
%
%
%
%
%

\bibliographystyle{model1-num-names}
\bibliography{scalar}

\end{document}